\documentclass{article}
\usepackage{amsmath}
\usepackage{amsfonts,amssymb}
\usepackage{graphics}
\usepackage{amsthm}

\usepackage{color}

\newtheorem{theorem}{Theorem}

\newtheorem{lemma}{Lemma}

\newtheorem{Remark}{Remark}

\def\cosh{\mbox{{\rm cosh}}}
\def\exp{\mbox{{\rm exp}}}
\def\div{\mbox{{\rm div}}}

\def\q{{\bf q}}
\def\p{{\bf p}}
\def\u{{\bf u}}

\def\R{\mathbb{R}}

\def\eps{\epsilon}
\def\ln{\mbox{{\rm ln}}}
\def\log{\mbox{{\rm ln}}}
\def\#{{{\cal D}_h}}
\def\equis{{\cal X}}

\title{On the relativistic BGK-Boltzmann model: asymptotics and hydrodynamics}

\author{A. Bellouquid\footnote{University Cadi Ayyad.
Ecole Nationale des Sciences Appliqu\'ees, Safi, Maroc. bellouquid@gmail.com}, J. Calvo\footnote{Dept. de Tecnologies de la Informaci\'o i les Comunicacions, Universitat Pompeu Fabra, Barcelona, Spain. juan.calvo@upf.edu}, 
J. Nieto, J. Soler\footnote{Dept. de Matem\'atica Aplicada, Universidad de Granada, Spain. jmnieto@ugr.es, jsoler@ugr.es} }

\begin{document}

 \maketitle

\begin{abstract}
The generalization of the BGK relaxation model to the special relativity setting is revisited here. We deal
with several issues related to this relativistic kinetic model which seem to have been overlooked in the previous physical literature, including the unique
determination of associated physical parameters,  classical, ultra-relativistic and hydrodynamical limits,
maximum entropy principles and the analysis of the linearized operator.
\end{abstract}

\section{Introduction}

Our aim is to give a mathematical description of a  gas in certain relaxation regimes. This will be done in terms of the relativistic BGK equation.
In that framework, such gases are regarded as consisting of many
microscopic structureless particles; our description will be given in terms of the relativistic
kinetic phase density. This  object allows us to compute tensorial moments giving the local macroscopic
physical quantities of the gas (particle density, pressure, etc).  We describe the construction of the relativistic BGK system, analyze
the hyperbolic hydrodynamical limits towards the relativistic Euler equations and study the well-posedness of the relativistic  linearized BGK
system, together with classical and ultra-relativistic limits.

The models of kinetic theory describe the time evolution of a collection of particles.
The main focus of kinetic theory has been on classical particles  and on the
special relativistic framework, but also there are
mathematical advantages in using kinetic models in general relativity instead of
their hydrodynamic counterparts, especially from the point of view of the ensuing singularities \cite{And,Crlibro,G}.  A few years before Einstein
established the principles of general relativity in 1916, J\"uttner \cite{juttner-primero} gave in 1911 the first step in extending the kinetic theory of gases to the
relativistic context, proposing a generalization of the Maxwellian distribution function which is widely accepted nowadays. The next step was due
to Walker \cite{walker} who derived the relativistic evolution equation in the absence of molecular collisions in 1935. Lichnerowicz and Marrot
\cite{liche} provide a first complete relativistic generalization of the Boltzmann equation (which was proposed in 1872) in 1940.  The analysis
of the Cauchy problem in relativistic kinetic theory was first addressed by Choquet--Bruhat
\cite{cho,cho-marsden,cho2}, and subsequently in a
wide literature, see for example \cite{glassey,Glas95}. The analysis of the linearized relativistic Boltzmann equation was first made in \cite{dud},
see also \cite{Drange}. The Cauchy problem for the relativistic Boltzmann equation has been recently revisited by Strain \cite{Strain}; he shows
that an initial datum starting close enough to a global equilibrium launches a unique global solution which decays with any polynomial rate towards that global equilibrium.

Some hydrodynamic extensions  of relativistic models have been obtained either using the Chapman-Enskog method or the moment method of
Maxwell and Grad  \cite{grad}, see Israel \cite{israel,israel-stewart}, De Groot {\it et al}  \cite{deGroot}, Marle \cite{Marle3} Degond {\it et al}
\cite{degond} and the references therein. We highlight here that the viscous (parabolic) limits frequently lead to an infinite wavefront
speed for the transport processes, which bear some inconsistency in a relativistic context, where the propagation speed cannot overcome that of
light \cite{Cercignani}. The hyperbolic asymptotic limits have the advantage of eliminating this spurious infinite velocity of the wave front
propagation and that makes them in general more suitable than parabolic limits in relativistic dynamics. In \cite{speck-strain} the local-in-time
hydrodynamic limit of the relativistic Boltzmann equation is analyzed using a Hilbert expansion and following the ideas developed in the classical
case \cite{ca1,ca2}. As a result, local solutions to the relativistic Boltzmann equation near the local relativistic Maxwellian are constructed, using
a class of solutions to
the relativistic Euler equations that show up in the hydrodynamic limit.  In this hydrodynamical context, one of the goals of our work will be to deduce the hyperbolic asymptotic in
terms of the different parameters involved in the relativistic BGK equation.

The BGK (Bhatnagar, Gross and Krook) model \cite{BGK},
proposed in 1954 for classical particles, became the most important model to solve the integro-differen\-tial
Boltzmann equation. The non-linear
quadratic (binary) collision term of the Boltzmann equation is replaced in the BGK model by a seemingly simpler term (however the nonlinearity
becomes exponential), which
makes the derivation of the transport equations for macroscopic variables easier. A problem which can be naturally addressed using the BGK
model is that of the relaxation of a state
of a fluid to equilibrium.

Our description --which is essentially based on the Marle model \cite{Marle1,Marle2}-- will come in terms of a relativistic extension of the hitherto known as
BGK model
$$
\frac{\partial f}{\partial t} + v \cdot \nabla_x f=
\nu (M_f-f),
$$
being $M_f$ a local Maxwellian implicitly defined by the requirement of having the same moments as the distribution function $f$, which
depends on time $t$, position $x$ and velocity $v$. From the physical point of view, the density of particles is assumed to
converge to an equilibrium represented by a Maxwellian function of the
velocity $v$ when the time $t$ becomes large. During the last 20 years, the BGK model has also found an
important application: the derivation of numerical schemes, namely kinetic schemes to
solve hyperbolic conservation laws, see \cite{PerthameB}. Some approaches to the relativistic BGK models have been proposed in the literature,
see for example \cite{brey,CercignaniKremer,Majorana,Marle1,Marle2} and the references therein.
It is also remarkable the applicability of this type of models in gravitation in order to analyze dispersion properties or the stability of special configurations such as galaxies. In fact, the BGK-equilibria (or BGK waves in 1D) have been analyzed for plasmas and gravitation  as stationary solutions of coupled systems of type  Vlasov, Vlasov-Boltzmann, Einstein-Vlasov, Landau-Boltzmann, Boltzmann-Maxwell, \ldots, see \cite{AW,And,CercignaniKremer,CH,CSS,CCSS,GA1,GA2,Hay,Ste} an the references therein.

The contents of the paper are as follows:
Physically meaningful quantities and local equilibria are introduced in the first Section. %
A generalization of the classical BGK model is introduced then in Section \ref{BtoBGK}, whose behavior in the non-
relativistic limit we analyze in Section \ref{ClassicalLimit}. Then we elaborate on the conservation laws of the model and how do they give rise to the relativistic and ultra-
relativistic Euler equations in some limiting regimes, being this the content of Sections \ref{relEuler} and \ref{UltrarelEuler}. The remaining
Sections \ref{displaylinearization} and \ref{existencelinearization} proceed with the analysis of the linearization of the
relativistic BGK operator, culminating with a global existence proof for the linearized BGK model. Several computations concerning moments of J\"uttner distributions and other technicalities are displayed in an Appendix.


\subsection{Some notes about special relativity}
The fundamental physical theory involved in this description that we want to set is special relativity. Just before detailing its precise role we will
introduce the related notations and conventions that we will follow during the exposition. Some background on special relativity and relativistic
kinetic theory can be found for instance in \cite{And,CercignaniKremer,deGroot,Landau,Liboff}.

The space-time coordinates in the four-dimensional Minkowsky's space are $x^\mu$, $\mu = 0, 1, 2, 3$, with
$x^0 =ct$ for the time and $x^1, x^2, x^3 $ for the position; here $c$ is the speed of light. The metric tensor $g_{\mu \nu}$ and its inverse $ g^{\mu \nu}$ are given by
$$
g_{\mu \nu}= g^{\mu \nu}=1 \  if \quad \mu=\nu =0, \quad -1 \quad if \quad \mu=\nu=
1,2,3 \quad and  \quad 0 \quad if \quad \nu \neq \mu.
$$
Our greek indices run from $0$ to $3$ and our latin (spatial) indices do from $1$ to $3$.
We will use Einstein's summation convention,
meaning that any index that appears twice in an expression,
is understood to be summed over its whole range. With the aid of the metric tensor we can  perform the operations of \emph{raising} and \emph{lowering} indices. That is, for any four-dimensional vector $u$ (four-vector hereafter),
$$
  g_{\alpha \nu}u^\nu = u_\alpha \quad \mbox{and}\quad g^{\alpha \nu}u_\nu = u^\alpha.
$$
This works in the same way for general tensor objects.

We will also consider vectors in the Euclidean three-dimensional space, which we will always denote by bold characters. Then, notations like $|
\u| = \sqrt{\u_i \u^i}$ and $\u \cdot \q = \u_i \q^i$ stand for the euclidean norm and scalar product respectively.


\subsection{Microscopic and macroscopic quantities}
We are going to describe the microscopic state of a relativistic gas by means of kinetic theory. For that, let us introduce the relativistic phase
density $f(t,x,\q) \ge 0$, which represents the density of particles with given spacetime position $(t,x)$ and momentum $\q \in \R^3$. We will
consider that all the gas particles
 have the same mass $m$.  Then
 the energy-momentum four-vector is defined as
$$
q^\mu = (q^0, c \q), \quad q^0 := c \sqrt{(mc)^2 + |\q|^2}.
$$
Note that $q^\mu$ has energy dimensions, whereas $\q$ has momentum units.

Next we can define some macroscopic moments and the entropy four-vector.  All these make sense at a given point $x^\mu$ as soon as $f(t,x,\cdot)\ge 0$ is not identically zero. The fact that the proper volume element $d\q /q^0$ is invariant
with respect to
Lorentz transformations (i.e. isometries of the Minkowsky space) \cite{CercignaniKremer} is a key physical feature of these definitions.
\begin{enumerate}
\item Particle-density four-vector
\begin{equation}
\nonumber
N^\mu(t,x)= \int_{\R^3} q^\mu f(t,x,\q)
\frac{d\q}{q^0},
\end{equation}
\item Energy-momentum tensor
\begin{equation}
\nonumber
T^{\mu \nu}(t,x)= \frac{1}{m} \int_{\R^3} q^\mu q^\nu f(t,x,\q)
\frac{d\q}{q^0},
\end{equation}
\item Entropy four-vector
\begin{equation}
\nonumber
S^\mu(t,x)= -\frac{k_B}{m} \int_{\R^3} q^\mu f(t,x,\q) \ln \left(\frac{f(t,x,\q)}{\eta}\right)
\frac{d\q}{q^0},
\end{equation}
being $k_B$ the Boltzmann's constant
 and $\eta= m/ \hbar^3$,
with $\hbar$ the Planck constant.
\end{enumerate}

\medskip
We may use the macroscopic moments $N^\mu, T^{\mu \nu}$ of
the relativistic phase density $f$ in order to compute other
macroscopic quantities of the gas, namely
\begin{enumerate}
\item The proper particle density $n_f$
\begin{equation}
\label{nefe}
n_f=\sqrt{N^\mu N_\mu}.
\end{equation}
\item The dimensionless velocity four-vector $u_f$
\begin{equation}
\label{uefe}
n_fu_f^\mu = N^\mu.
\end{equation}
Note that $u_f^\mu (u_f)_\mu = 1$ and then $u^\mu =(\sqrt{1+|\u|^2},\u)$. Thus for $\mu =0$,  we can deduce the relation
\begin{equation}
\label{rel1}
n_f \sqrt{1+|\u_f|^2}= \int_{\R^3} f(t,x,\q) d\q.
\end{equation}
\item The proper energy density $e_f$
\begin{equation}
\nonumber
 e_f = (u_f)_\mu (u_f)_\nu  T^{\mu \nu}.
\end{equation}
\item The proper pressure $p_f$
\begin{equation}
\label{temperatura}
 p_f = \frac13 ((u_f)_\mu (u_f)_\nu  - g_{\mu \nu})T^{\mu \nu}.
\end{equation}

\item The proper entropy density $\sigma_f$
\begin{equation}
\nonumber
 \sigma_{f}= S^\mu (u_f)_\mu.
\end{equation}

\end{enumerate}


\subsection{The J\"uttner equilibrium}
\label{equilibrio}
The generalization of the classical global Max\-wellian to this setting is the so-called J\"uttner equilibrium (or relativistic Max\-we\-llian) \cite{juttner-primero,Juttner}. The J\"uttner distribution can be regarded as a function $J(n,\beta,\u;\q)$ describing the state of
a gas in equilibrium, depending on five parameters: $n\ge 0$, $\beta > 0$ and $\u \in \R^3$. It is given by the formula
\begin{equation}
\nonumber
 J(n,\beta,\u;\q) = \frac{n}{(mc)^3 M(\beta)} \exp \left\{- \frac{\beta}{m c^2}
u_\mu q^\muÊ\right\} \end{equation} or equivalently
$$
J(n,\beta,\u;\q) = \frac{n}{(mc)^3 M(\beta)} \exp \left\{- \frac{\beta}{m c} \left( \sqrt{1+|\u|^2}\sqrt{(m c)^2 + |\q|^2} -  \u \cdot \q \right)Ê\right\}.
$$
Since $J(n,\beta,\u;\q)$ is thought of as an equilibrium distribution, then $n$ is interpreted as its particle density, $\u $ as the spatial part of
the four-velocity $u$ (and as such $u_\mu u^\mu = 1$) and $m c^2 /(k_B \beta)$ as the equilibrium temperature. Note also that $\beta$ is
dimensionless as $ M(\beta)$ is, which takes the form
\begin{eqnarray}
\label{eme}
 M(\beta)&=& \frac{1}{(mc)^3}\int_{\R^3}\exp\left\{ -\frac{\beta}{m c} \sqrt{(mc)^2+|\q|^2}\right\}d\q \nonumber \\
& =&\int_{\R^3}\exp\left\{ -\beta \sqrt{1+|\p|^2}\right\}d
\p.
\end{eqnarray}
In such a way, we have that
\begin{equation}
  n u^\mu = \int_{\R^3} q^\mu J(n,\beta,\u;\q) \frac{d\q}{q^0}.
  \label{nuJ}
\end{equation}
 As regards the size of the exponent in the above formulae, we point out that
\begin{equation}
\label{scalaruq}
u_\mu q^\mu \ge m c^2
\end{equation}
which is an straightforward consequence of the Cauchy--Swartz inequality for timelike vectors in
Minkowski space.

\section{From relativistic Boltzmann to relativistic BGK}
\label{BtoBGK}
The definitions from the previous section apply to any relativistic phase density $f=f(t,x,\q)$, which, in case that no forces are involved, should obey a kinetic equation of the following
form
\begin{equation}
\label{paco}
q^\mu \frac{\partial f}{\partial x^\mu}= C(f).
\end{equation}
We have a transport part on the left-hand side and a collision part
$C(f)$ on the right-hand side, as in the non-relativistic setting. In the simplest case
particles interact only by means of elastic collisions --no other forces are assumed to play a role--. Then
 $C(f)$ should be determined in such a way that the conservation laws
for particle number, energy and momentum hold:
\begin{equation}
\label{conservationlaws}
 \frac{\partial N^\mu }{ \partial x^\mu}=0, \quad \frac{\partial T^{\mu \nu} }{ \partial
x^\nu}=0.
\end{equation}

\noindent The relativistic Boltzmann equation is a paramount example
of such kinetic models. It arises as a particular case of \eqref{paco} with $C(f) = Q(f,f)$,
which stands for a bilinear operator in $f$, verifying that $Q(J,J)=0$.
The detailed description of the relativistic collision operator $Q(f,f)$  is inessential for our present purposes; we refer the reader to
\cite{And,CercignaniKremer,deGroot} for more information of the subject. We just point out that the relativistic extension of the BGK model and
related descriptions --see \cite{Majorana} for instance-- arise as model equations for the relativistic Boltzmann description.


\subsection{The BGK Model}

Now we are going to construct a BGK-type model which satisfies the
conservation laws (\ref{conservationlaws}). In this way we will recover a model previously posed by Marle \cite{Marle1,Marle2}. Making a parallel with the relativistic Boltzmann equation and the classical BGK model, we want to build an equation for $f$ of the following form
\begin{equation}
\label{relBGK}
\partial_t f+c\hat{q}\cdot\nabla_x f= \frac{mc^2\omega}{q^0
} (J_f - f):= \frac{c\, Q_{BGK}^R(f)}{q^0
},
\end{equation}
the right hand side being a relaxation operator constructed from a
local relativistic Maxwellian $J_f$, which we will detail below. Here  $\omega$ denotes the collision frequency and $\hat{q}= \q/\sqrt{(mc)^2+|\q|^2}$.
Note that such a model has also the form \eqref{paco}, with $C(f) =Q_{BGK}^R(f)= mc\omega (J_f - f)$.

We define $J_f$ implicitly as a local J\"uttner distribution $J(n,\beta,\u;\q) $
whose parameters $n, \u, \beta$ will be given in the next Lemma \ref{Bellouq} through some integrals of $f$ in an unique way.
In order to do that, we impose that  functions verifying \eqref{relBGK} should also satisfy the conservation laws \eqref{conservationlaws}. That requirement turns up to be
(via multiplication by $1$ and $q^\mu/m$ and integration in \eqref{relBGK}) equivalent to
\begin{equation}
\label{veintinueve}
\int_{\R^3}  J_f \frac{d\q}{q^0
} =\int_{\R^3}  f \frac{d\q}{q^0},
\end{equation}
\begin{equation}
\label{veintiocho}
\int_{\R^3} q^\mu J_f \frac{d\q}{ q^0}= \int_{\R^3} q^\mu f \frac{d\q}{q^0}.
\end{equation}
Actually, let us see that these five conditions allow us to determine the five parameters of the J\"utner equilibrium
and to construct the relaxation operator.
\begin{lemma}
\label{Bellouq}
 Let $f=f(\q) \geq 0$, for $\q\in
\R^3$,  a non-zero phase density (in the almost everywhere sense)
for which the moments $N^\mu, T^{\mu \nu}$ exist. Then, we can find a unique set of quantities $n$, $\u$,  $\beta$ and a function $ J_f:=J(n,
\beta,\u;\q) $ such that \eqref{veintinueve} and \eqref{veintiocho}  are fulfilled. Moreover, for this choice of $J_f$, a relativistic phase density $f$
verifying \eqref{relBGK} fulfills the conservation laws \eqref{conservationlaws}.
\end{lemma}
\begin{proof}
We first note that taking $n=n_f$ and $\u = \u_f$, \eqref{veintiocho} is a direct consequence of  \eqref{uefe} and \eqref{nuJ}.

It only remains to determine the parameter $\beta$ in order that (\ref{veintinueve}) holds. Following Lemma \ref{lema3}, condition
(\ref{veintinueve}) is equivalent to the equation
\begin{equation}
\label{betadef}
   \frac{K_1(\beta)}{K_2(\beta)} = \frac{ mc^2}{n_f} \displaystyle \int_{\R^3}  f \frac{d\q}{q^0
},
\end{equation}
The function on the left hand side is known to be increasing and with range $[0,1)$, see Section \ref{mono} in the Appendix and \cite{librotablas}.
It only remains to prove that the right hand side is in the same range. In order to do that we use the invariance of $n_f$ and $\frac{d\q}{q^0}$
under Lorentz transformations (see Lemma  \ref{lorentz} in the Appendix)  to write
\[
\frac{\displaystyle \int_{\R^3}  f \frac{d\q}{q^0}}{n_f} =\frac{\displaystyle \int_{\R^3}  f_\Lambda \frac{d\q}{q^0
}}{\displaystyle \int_{\R^3}  f_{\Lambda} \ d\q},
\]
where we have considered the  Lorentz boost $\Lambda$ such that $u_{f_\Lambda} = (1,0,0,0)$. We can conclude by noticing that $q^0=c
\sqrt{(mc)^2+|\q|^2} \geq mc^2$.
\end{proof}

For future reference we denote $\beta_f$ the unique parameter $\beta$ verifying equation \eqref{betadef}. We also remark here (see Lemma \ref{lorentz} in the Appendix) that
this quantity is Lorentz invariant.
So finally the BGK-type model that we propose is given by (\ref{relBGK}) together with
\[
J_f=J(n_f,\beta_f,\u_f;\q) .
\]
The well-posedness of (\ref{relBGK}) remains as an open problem (compare with
\cite{Perthame,Bellouquid1}); see however
\cite{Perthameperdido} for the 2-dimensional case.

\subsection{Conservation laws}
Under no external forces and no radiation, any matter model should  obey the constitutive relations \eqref {conservationlaws}. For our model,  these conservation laws can be obtained from \eqref{relBGK} upon  multiplication by $ (1,q^\mu/m) $ and integration in $\R^3$  against $\frac{d\q}{q^0
}$, by noticing that Lemma \ref{Bellouq} allow us to  ensure that the  corresponding integrals of the rhs are zero. Then, the conservation laws read: first,
\begin{equation}
\label{Mass}
\frac{\partial}{\partial t} \int_{\R^3} f d\q \ +c\,  \div_x \int_ {\R^3} \frac{\q f}{\sqrt{(mc)^2 + |\q|^2}}\ d\q = 0,
\end{equation}
as the equation for $N^\mu$ (conservation of the number of particles) which, using \eqref{uefe} and \eqref{rel1}, leads to
\begin{equation*}
\label{Mass1}
   \frac{\partial}{\partial t} \left( n_f \sqrt{1+\u_f^2} \right)  \ +c\,  \div_x (n_f \u_f) = 0.
\end{equation*}
Second, the equation for $T^{0\mu}$ (conservation of energy --zeroth component  of the four--momentum) takes the form
\begin{equation}
\label{energy}
   \frac{1}{m} \frac{\partial}{\partial t} \int_{\R^3} \sqrt{(mc)^2  + |\q|^2} f\ d\q \ + \frac{c}{m}\div_x \int_{\R^3} \q f d\q =
   0.
\end{equation}
And third, the equations for $T^{i\mu}$, which represents the conservation of the spatial part of the  four--momentum, is
\begin{equation}
\label{momentum}
   \frac{c}{m} \frac{\partial}{\partial t} \int_{\R^3} \q^i f d\q \  + \frac{c^2}{m} \frac{\partial }{\partial x_j} \int_{\R^3}
    \frac{\q^i \q^j f}{\sqrt{(mc)^2 + |\q|^2}}\ d\q = 0.
\end{equation}

\subsection{ The H-Theorem and the entropy}
\label{htheorem}

An important property of the BGK model (\ref{relBGK}) is the $H$- theorem,
which gives the local dissipation law for the functional $\int_{\R^3}  f \ln(f)\
d\q$, namely, the following equation:
\begin{equation}
\label{entropy21}
  \frac{k_B}{m} \partial_t \int_{\R^3} f \ln(f/\eta)\
d\q + \frac{c k_B}{m}\div_x \int_{\R^3} \q f \ln(f/\eta)\frac{d\q}{ q^0
} \leq 0.
 \end{equation}
 In order to achieve the $H$-Theorem, it suffices to show
the dissipation property
\begin{equation}
\label{entropydiss}
\frac{k_B}{m} \int_{\R^3} \ln(f/\eta) Q_{BGK}^R(f) \frac{d\q}{
q^0
} \leq 0.
 \end{equation}
 This is already proved in \cite[Chapter 8, eq.  (8.8)]{CercignaniKremer}.

Let us see next that the local J\"uttner equilibria arising in our  framework are determined by an extremality principle that is akin to  that of maximum entropy.

\begin{lemma}[Minimum free energy principle]
\label{maximumPrinciple}
Let $f(\q)\ge 0$ be given and let $J_f$ be the  associated J\"uttner equilibrium constructed in Lemma \ref {Bellouq}. Then,  the following inequality
$$
\left(\sigma-\frac{k_B \beta}{(mc)^2} e\right)_{J_f} - \left(\sigma- \frac{k_B \beta}{(mc)^2} e\right)_f \ge 0
$$
holds.
\end{lemma}
\begin{proof}
The proof is similar to that of Proposition 2.1 in \cite{Kunik2004}.  Using a Taylor development, we get the identity
$$
f \log (f/\eta) = J_f \log (J_f/\eta) + (f- J_f)[\log (J_f/\eta) + 1]  + \frac{(f-J_f)^2}{2 \xi},
$$
being $\xi$ some momentum distribution that can be written as a convex combination of $f$ and $J_f$. Then, taking in to account Lemma \ref{Bellouq}, the  difference between the entropies associated with $J_f$ and $f$ can be  written as
\begin{eqnarray*}
  \sigma_{J_f} - \sigma_f  &=& \frac{k_B}{m} (\u_f)_\mu \left[ \int_ {\R^3} q^\mu \log (J_f/\eta) (f-J_f)\frac{d\q}{q^0} + \int_{\R^3} q^ \mu \frac{(f-J_f)^2}{2 \xi} \frac{d\q}{q^0} \right]
\\
 &=& \frac{k_B}{m} \log \left( \frac{n_f}{\eta (mc)^3 M(\beta_f)} \right) (\u_f)_\mu \int_{\R^3} q^\mu (f-J_f) \frac{d\q}{q^0}
\\
 &&- \frac{k_B}{m} \frac{\beta_f}{mc^2} (\u_f)_\mu (\u_f)_\nu  \int_ {\R^3} q^\mu q^\nu (f-J_f) \frac{d\q}{q^0}
 \\
 &&+ \frac{k_B}{m} (\u_f)_\mu  \int_{\R^3} q^\mu \frac{(f-J_f)^2}{2 \xi} \frac{d\q}{q^0}.
\end{eqnarray*}
Taking into account \eqref{scalaruq} we deduce that the last term  above is non-negative. The first term is already zero thanks to Lemma  \ref{Bellouq}. Next we use that $\beta_{J_f} = \beta_f$ to write
$$
 \sigma_{J_f} - \sigma_f \ge  - \frac{k_B}{m} \frac{\beta_f}{mc^2}  (\u_f)_\mu (\u_f)_\nu  \int_{\R^3} q^\mu q^\nu (f-J_f) \frac{d\q} {q^0} = -\frac{k_B}{(mc)^2} (\beta_f e_f - \beta_{J_f}e_{J_f}).
$$
Altogether the result follows.
\end{proof}


\section{Asymptotic limits of the relativistic BGK model}
\label{dimension}
One of the aim of the paper is to perform three distinguished limits of the BGK system \eqref{relBGK}: the non relativistic limit, the hydrodynamical limit and the ultra relativistic limit. In order to do that we will consider the physical meaning of the involved constants represented on a set of only four typical variables, which allow to rewrite \eqref{relBGK} in a dimensionless form for its mathematical study. We will precise later the connections between these four variables which give rise to the three  limits. Let us choose a scaling such that:
\[
\q= \bar\q \mu, \quad x= \bar x L, \quad t = \bar t \tau,
\]
where  ``$\bar{\ }$"{} stands for dimensionless numbers and $\mu$, $L$ and $\tau$ stand for the typical microscopic momentum, the typical length and the typical time, respectively.  Also, we perform the change of unknown
\[
 f(t,x,\q)  = \frac{{\cal N}}{\mu^3} \bar f (\bar t, \bar x, \bar \q),
\]
where ${\cal N}$ is the typical density. Then,  system \eqref{relBGK} becomes  now
\begin{equation}
\label{barBGK}
\frac{\partial }{\partial \bar t} \bar f+\Big( \frac{\tau \mu }{m L }\Big)  \frac{ \bar\q \cdot \nabla_{\bar x} \bar f }{\sqrt{1+ \big| \frac{\mu }{mc}  \bar \q  \big|^2}}=( \omega \tau )  \frac{1}{\sqrt{1+ \big| \frac{\mu }{mc}  \bar \q  \big|^2}}  \Big(\bar J_{\bar f} - \bar f\Big),
\end{equation}
where $\bar J_{\bar f}= \frac{\mu^3}{\cal N}  J_{f} $ must  be specified for each limit. Also, \eqref{entropy21} becomes
\begin{equation}
\label{entropyscaled}
\partial_{\bar t} \int_{\R^3} \bar f\,  \ln\Big(\frac{{\cal N} {\bar f}}{\mu^3\eta}\Big)\,
d\bar \q +\Big( \frac{\tau \mu }{m L }\Big) \div_{\bar x} \int_{\R^3} \bar \q \bar f \,  \ln\Big(\frac{{\cal N} {\bar f}}{\mu^3\eta}\Big)\, \frac{d \bar \q }{\sqrt{1+ \big| \frac{\mu }{mc}  \bar \q  \big|^2}} \leq 0.
 \end{equation}

Then, the three  limits will be attained by choosing the following relations:
\begin{equation}
\label{3escalas}
\begin{tabular}{|l|c|c|c|c|}\hline
Limit & \multicolumn{3}{|c|}{Scaling relations} & Asymptotics \\ \hline
 & & & & \\[-8 pt]
Non relativistic & $\mu =\displaystyle\frac{mL}{\tau} $ & $\displaystyle\omega \tau =\, $cst  &
 $\displaystyle\frac{\eta \mu^3}{{\cal N}}= \, $cst & $\mu < < mc$ \\[7 pt] \hline
  & & & & \\[-8 pt]
Hydrodynamic & $\mu =\displaystyle\frac{mL}{\tau} $ & $c =\displaystyle\frac{L}{\tau} $  &
 $\displaystyle\frac{\eta \mu^3}{{\cal N}}= \, $cst & $\displaystyle\frac{1}{\omega} < < \tau $ \\[7 pt] \hline
 & & & & \\[-8 pt]
Ultra relativistic &  $\omega \tau=\, $cst  & $c =\displaystyle\frac{L}{\tau} $ &
 $\displaystyle\frac{\eta \mu^3}{{\cal N}}= \, $cst & $ mc < < \mu $ \\[7 pt] \hline
\end{tabular}
\end{equation}
Also, a dimensionless relativistic BGK model will be analyzed in Sections \ref{displaylinearization} and \ref{existencelinearization} when the
nondimensional constants involved in \eqref{barBGK}--\eqref{entropyscaled} are 1, which corresponds with the choice
\begin{equation}
\label{normalscale}
\begin{tabular}{|c|c|c|c|}\hline
\multicolumn{4}{|c|}{Normalized scaling relations}  \\ \hline
$c =\displaystyle\frac{L}{\tau} $ & $\omega \tau = 1$ &
 $\displaystyle\frac{\eta \mu^3}{{\cal N}}= 1$& $\mu = mc$ \\[7 pt] \hline\end{tabular}
\end{equation}


\subsection{The non relativistic limit}
\label{ClassicalLimit}
In order to  give consistency to our approach to the  relativistic BGK model, let us recover in the non-relativistic limit the classical BGK system, see also  \cite{Calogero2004,Cercignani,Strain2} for a similar analysis in this context.

In order to perform the non relativistic limit we choose the first scaling given in \eqref{3escalas}.
It corresponds to the assumption that the typical microscopic momentum $\mu$ is small compared with $mc$, i.e., the speed of particles is slow compared with the speed of light. Then, setting $\bar \omega := \omega \tau $ and $\bar \eta :=\eta \mu^3/{\cal N}$  --assumed to be constants--  and defining $\epsilon := \mu / (mc)< < 1$, our system \eqref{barBGK} becomes  now
\begin{equation}
\label{relBGKscalled}
\partial_{\bar t}  f_{nr}+ \frac{1}{\sqrt{1+|\epsilon \bar\q |^2}} \bar\q \cdot \nabla_{\bar x}  f_{nr}= \frac{\bar \omega }{\sqrt{1+|\epsilon \bar\q |^2}}
\Big(\frac{\mu^3}{\cal N}  J_{f} -  f_{nr}\Big),
\end{equation}
where it only remains to write in a compact form the expression of the right hand side. To do that we write:
\begin{eqnarray}
\frac{1}{q^0} &=& \frac{1}{mc^2}\displaystyle\frac{1}{\sqrt{1 + \epsilon^2 |\bar\q|^2}} =\frac{1}{mc^2}\left( 1 - \epsilon^2 \displaystyle\frac{|\bar\q|^2}{2}+{\cal O}(\epsilon^4)\right), \label{q0called}
\end{eqnarray}
and then, after the change of variables $d\q = \mu^3 d\bar \q$ on the integrals, we get
\begin{eqnarray}
\hspace{-0.9 cm}  n_f &= &  \left[ \left( \int_{\R^3} f(t,x,\q) d\q \right)^2 -  \sum_{i=1}^3  \left( \int_{\R^3} c{\q}_i f(t,x,\q) \frac{d\q}{q^0
} \right)^2 \right]^{\frac12} \nonumber
\\
&=&  {\cal N}\!\left[ \left( \int_{\R^3}  f_{nr}(\bar t,\bar x,\bar\q) d\bar\q \right)^2\! \!  - \epsilon^2 \sum_{i=1}^3  \left( \int_{\R^3} \bar{\q}_i  f_{nr}(\bar t,\bar x,\bar\q) d\bar\q
\right)^2 \right]^{\frac12}\!\!\! + {\cal O}(\epsilon^4).
 \label{enescalled1}
\end{eqnarray}
Defining then the classical (scaled) density as usual,
\[
n_{nr} (\bar t, \bar x)= \int_{\R^3}  f_{nr} (\bar t,\bar x,\bar\q) d\bar{\q},
\]
we obtain the asymptotic behavior of $n_f$
\begin{equation}
n_f(t,x)=  {\cal N}  n_{nr} + {\cal O}(\epsilon).
\label{enescalled}
\end{equation}
In a similar way, we use \eqref{q0called} to   find
\[
n_f\u_f = \int_{\R^3} c \q f(t,x,\q) \frac{d\q}{q^0}=
  \frac{{\cal N} \mu}{mc} \int_{\R^3} \frac{\bar\q     f_{nr}(\bar t,\bar x,\bar\q) }{\sqrt{1 + \epsilon^2 |\bar\q|^2}} d \bar \q = \epsilon  {\cal N}\! \!
 \int_{\R^3} \bar\q     f_{nr}\,  d \bar \q + {\cal O}(\epsilon^2)
\]
which in combination with (\ref{enescalled}) allows to get
\begin{eqnarray}
\label{uscalled}
\u_f = \epsilon \left(\frac{1}{n_{nr}}  \int_{\R^3} \bar\q  f_{nr}(\bar t,\bar x,\bar\q) d\bar\q \right) + {\cal O}(\epsilon^2) := \epsilon   \u_{nr} + {\cal O}(\epsilon^2),
\end{eqnarray}
where $  \u_{nr}$ is actually the (scaled) classical velocity. In order to obtain the expansion for $\beta_f$ we first combine a Taylor expansion of $1/\sqrt{1-s^2}$ around $s=0$ together with \eqref{enescalled1} to obtain
\begin{equation}
\frac{1}{n_f} = \frac{1}{{\cal N} n_{nr}}  \Big( 1+ \frac{\epsilon^2}{2} | \u_{nr}|^2\Big)+ {\cal O}(\epsilon^4).
\label{1sobren}
\end{equation}
Second, we use \eqref{q0called} to obtain
\[
 \int_{\R^3} f(t,x,\q) \frac{d\q}{q^0}=
 \frac{ {\cal N} }{mc^2} \left( \int_{\R^3}   f_{nr}(\bar t,\bar x,\bar\q) d \bar \q   - \epsilon^2 \int_{\R^3} \frac{|\bar \q|^2}{2}  f_{nr}(\bar t,\bar x,\bar\q) d \bar \q \right)  + {\cal O}(\epsilon^4),
\]
which, combined with \eqref{1sobren} and using the definition  \eqref{betadef} of $\beta_f$, allow us to write
\[
\displaystyle\frac{K_1}{K_2}(\beta_f) =1 + \frac{\epsilon^2}{2}\left(|  \u_{nr}|^2-  \displaystyle\frac{1}{n_{nr}}\displaystyle
\int_{\R^3} |\bar\q|^2  f_{nr}(\bar t,\bar x,\bar\q) d\bar\q \right)
+  {\cal O}(\epsilon^4).
\]
In particular, we observe that $K_1(\beta_f)/K_2 (\beta_f)$ approaches 1, so equivalently $\beta_f$ is large. This is precisely why we use the following expansion  (check \cite{CercignaniKremer}), valid for $\beta_f >>1$
\begin{eqnarray*}
\displaystyle\frac{K_1}{K_2}(\beta_f) = 1 - \displaystyle\frac{3}{2\beta_{f}} + {\cal O} \left(\displaystyle\frac{e^{-\beta_{f}}}{\beta_{f}}\right).
\end{eqnarray*}
Finally, we use the two expansions of $K_1/K_2$ to obtain
\[
\frac{1}{\beta_f}= \frac{\epsilon^2}{3}\left(\frac{1}{n_{nr}}  \int_{\R^3} | \bar\q|^2   f_{nr} (\bar t,\bar x,
\bar\q) d\bar\q -  | \u_{ nr}|^2 \right)+ {\cal O} \left(\displaystyle\frac{e^{-\beta_{f}}}{\beta_{f}}\right) +  {\cal O}(\epsilon^4),
\]
and, actually, define the (scaled) temperature at equilibrium,
\begin{equation}
 T_{nr} = \frac{1}{ \beta_{nr}} := \frac{1}{3}\left(\frac{1}{n_{nr}}  \int_{\R^3} | \bar\q|^2  f_{nr}(\bar t,\bar x,
\bar\q) d\bar\q -  |  \u_{nr}|^2 \right),
\label{scaledT}
\end{equation}
which corresponds with the classical one. With this notation we can also write the expansion for $\beta_f$
\begin{equation}
{\epsilon^2 \beta_{f} =  \beta_{nr} \left(1-  {\cal O}(\epsilon^2)\right)}.
\label{betascaled}
\end{equation}

Let us now check the behavior of the J\"uttner function in the non-relativistic limit. We first
observe that (see \cite{librotablas})
\begin{equation}
\frac{1}{M(\beta_f)} = \frac{e^{\beta_f}}{\left(\frac{2\pi}{\beta_f}\right)^{3/2}} \left( 1+ {\cal O} \Big(\frac{1}{\beta_f}\Big) \right) =
 \frac{ e^{\beta_f}}{ \epsilon^3} \left(\frac{2\pi}{ \beta_{nr}}\right)^{-3/2} \Big( 1+ {\cal O} (\epsilon^2) \Big) .
\label{emescaled}
\end{equation}
On the other hand, using  \eqref{uscalled},  we can get  the following equality  for the size of the exponent in the J\"uttner function
\begin{equation}
\frac{(u_f )_\mu q^\mu}{mc^2} =    \sqrt{1+|\u_f|^2}\sqrt{1+ \epsilon^2 |\bar \q|^2} - \epsilon \u_f \cdot \bar \q =1 +\frac{ \epsilon^2}{2} |\bar\q -   \u_{nr} |^2 + {\cal O}(\epsilon^3).
\label{uqscaled}
\end{equation}
Then, putting all the expansions \eqref{enescalled}, \eqref{uscalled}, \eqref{betascaled}, \eqref{emescaled} and \eqref{uqscaled} together, we finally obtain
\[
\frac{\mu^3}{\cal N}  J_{f} =  \displaystyle{n_{nr}}{\left(\displaystyle\frac{2\pi}{  \beta_{nr}}\right)^{-3/2}}\exp \left\{-
\beta_{ f_{nr}} \frac{|\bar\q -   \u_{nr} |^2}{2} \right\} + {\cal O}(\epsilon^{2}):=G({f_{nr}}) + {\cal O}(\epsilon^{2}),
\]
which is, up to order $\epsilon^{2}$, the classical Maxwellian $G(f_{nr})$, that can be also written as
\[
G(f_{nr})= \frac{ n_{nr} }{{ (2\pi T_{nr})^{3/2} } } \exp \left\{-  \frac{|\bar\q -
\u_{nr} |^2}{2 \,  T_{nr}} \right\}.
\]
Then, \eqref{relBGKscalled} becomes (up to order $\epsilon^2$) the classical BGK equation
\[
\partial_t f_{nr}+ {\bar{\q} }\cdot\nabla_x f_{nr}= \bar \omega (G(f_{nr}) - f_{nr}).
\]
We can conclude that, assuming the convergence of all the involved moments of $ f_{nr}$, the limit  satisfies the classical
BGK equation.

Next, for the energies, identities \eqref{q0called} and (\ref{uqscaled}) lead us to the following ex\-pre\-ssion
\[
e_f = {\cal N} c^2 \left( n_{nr} \Big(1-\frac{\epsilon^2}{2}|  \u_{nr}|^2\Big)+\frac{\epsilon^2}{2} \int_{\R^3} | \bar\q -   \u_{nr}|^2  f_{nr}
(\bar t,\bar x,\bar\q) d\bar\q \right)+ {\cal O} (\epsilon^3).
\]
The previous relationship can be seen as the usual expansion of the relativistic expression for
the kinetic energy for small velocities: the first term, ${\cal N} c^2 n_{nr}$, is the rest mass
energy of the system, while up to first of order in $\epsilon^2$, $\frac{{\cal N} c^2}{2} n_{nr} |  \u_{nr}|^2$ is the (Newtonian) kinetic energy of the center of mass and $\tfrac{1}{2}\int_{\R^3} | \bar\q -   \u_{nr}|^2  f_{nr} (\bar t,\bar x,\bar\q) d\bar\q$ the internal (Newtonian) kinetic energy
of the individual particles (with respect to the center of mass frame). The latter is a measure
for the temperature of the system, such that with (\ref{scaledT}) the above expression can be
written as
\begin{equation}
\label{restenergy}
e_f + \epsilon^2 \frac{{\cal N} c^2}{2} n_{nr} |  \u_{nr}|^2 =  {\cal N} c^2 \Big ( n_{nr} +  \frac32 \epsilon^2n_{nr}  T_{nr} \Big) + {\cal O}(\epsilon^{3}).
\end{equation}
Note that in the rest frame ($ \u_{nr}=0$) this expression reduces to the well-known relation
between the kinetic energy and the temperature of a monoatomic ideal gas (except for the rest
mass energy term).

Finally, we can also check the relation between the relativistic and classical  entropies. Using again the previous expressions for $u_f$,  we find,
\begin{equation}
\label{entropia}
\sigma_f = -\frac{k_B{\cal N}}{m} \int_{\R^3}   f_{nr}(\bar t,\bar x,\bar\q ) \ln \left(\frac{ f_{nr} (\bar t,\bar x,\bar\q)}{\bar \eta}\right) d\bar\q+{\cal O} (\epsilon^2),
\end{equation}
where the first term of the right hand side corresponds, up to order $\epsilon^2 $, to the classical entropy.  We also stress that, in the light of what was done in here, Lemma \ref{maximumPrinciple} reduces in the non  relativistic limit to the usual maximum entropy principle for the  Maxwell--Boltzmann distribution. This is a direct consequence of  relation (\ref{entropia}) and the fact that the local energy reduces  to the rest mass in the non relativistic limit --see (\ref {restenergy})--, but then both the density of $f$ and that of the  associated Maxwellian are equal. Thus $(\beta e)_{J_f}-(\beta e)_f$  cancels in that limit and the functional in Lemma  \ref {maximumPrinciple} converges to the classical entropy.

\subsection{Derivation of the relativistic Euler equations}
\label{relEuler}

This part is devoted to the connection between the relativistic BGK
model (\ref{relBGK}) and the relativistic Euler equations. These are obtained as a particular limit of the BGK model. Obtaining hydrodynamical hyperbolic limits from the  relativistic BGK model is one of the goals of this paper, but the  analysis given in the previous subsection allows to
connect this type of  limits in the relativistic setting with their counterpart limits for the classical BGK system. At least in the framework of regular
solutions, hyperbolic and non relativistic (Newtonian) limits should commute. Then, in this sense we can assure that the relativistic Euler equations
obtained in the next Section \ref{relEuler} are the relativistic counterpart to the classical ones derived from the classical BGK model.

In order to proceed with the hyperbolic hydrodynamic limit, we choose the second scaling given in \eqref{3escalas}.
It corresponds to a fluid  behavior:  we assume that the time between collisions  $1/\omega$ is small compared with the typical time $\tau$, and the fluid description takes place.  Then, if we take  now $\epsilon = \omega \tau$, we put  $\eta \mu^3/{\cal N}=1$ in order to abbreviate the expression for the entropy
and we eliminate the superscript ``$\bar{\ }$"{} to simplify the notation, we obtain the  following  fluid-scaled version of the relativistic BGK model:

\begin{equation}
\label{BGKscalEuler}
\partial_t f_{Eu}^{\epsilon}+\hat q   \cdot\nabla_x f_{Eu}^{\epsilon}= \frac{1}{\epsilon \,\sqrt{1+\q ^2}} (J_{f_{Eu}^{\epsilon}} - f_{Eu}^{\epsilon})\end{equation}
where now $\hat q = \q/\sqrt{1+\q ^2} $ and $J_{f_{Eu}^{\epsilon}}$ stands in dimensionless form:
\[
J_{f_{Eu}^{\epsilon}}= \frac{n_{f_{Eu}^{\epsilon}}}{M(\beta_{f_{Eu}^{\epsilon}})} \exp \left\{- \beta_{f_{Eu}^{\epsilon}} \left( \sqrt{1+\u_{f_{Eu}^{\epsilon}}^2}\sqrt{1 + |\q|^2} -  \u_{f_{Eu}^{\epsilon}} \cdot \q \right)Ê\right\}.
\]
Also, the conservation laws  (\ref{Mass}), (\ref{energy}) and (\ref{momentum}) read now
\begin{eqnarray}
\label{MassEuler} &&
\frac{\partial}{\partial t} \int_{\R^3} f_{Eu}^{\epsilon} d\q \ +  \div_x \int_ {\R^3} \frac{\q f_{Eu}^{\epsilon}}{\sqrt{1 + |\q|^2}}\ d\q = 0, \\
\label{energyEuler} &&
\frac{\partial}{\partial t} \int_{\R^3} \sqrt{1+ |\q|^2} f_{Eu}^{\epsilon} d\q \ +
\div_x \int_{\R^3} \q f_{Eu}^{\epsilon} d\q = 0, \\
\label{momentumEuler} &&
\frac{\partial}{\partial t} \int_{\R^3} \q^i f_{Eu}^{\epsilon} d\q \  + \frac{\partial}{\partial x_j}  \int_{\R^3}
    \frac{\q^i \q^j f_{Eu}^{\epsilon} }{\sqrt{1 + |\q|^2}}\ d\q = 0 .
\end{eqnarray}
Let us denote by $\{f_{Eu}^{\epsilon}(t, x, q) \}_{\epsilon \geq 0}$ a sequence of
nonnegative solutions of the equation (\ref{BGKscalEuler}) such that $f_{Eu}^{\epsilon}$
converges almost everywhere to a nonnegative function $f_{Eu}$, as
$\epsilon$ goes to zero. Assume also that all the moments converge
to their corresponding moments and  that $J_{f_{Eu}^{\epsilon}}$ converges
to $ J_{f_{Eu}}$ in an appropriate sense to be specified later. Then,
multiplying (\ref{BGKscalEuler}) by $\epsilon$ and letting $\epsilon$ go to $0$,
one has in the limit:
\begin{equation}
f_{Eu}= J_{f_{Eu}}.
\nonumber
\end{equation}
Then, deriving  macroscopic equations from (\ref{MassEuler}), (\ref{energyEuler}) and (\ref{momentumEuler}), will require information on different
moments of the equilibrium $J_{f_{Eu}}$, stated in Lemma \ref{lema3}. Under suitable convergence assumptions (which we prove below),
we can ensure that
\[
\int_{\R^3} A(\q) \, f_{Eu}^{\epsilon} \, d\q \  \to \
\int_{\R^3}  A(\q) \, f_{Eu} \, d\q = \int_{\R^3} A(\q)  \, J_{f_{Eu}} \, d\q,
\]
where $A(\q)$ stands for all the involved moments of $f_{Eu}$, that is,
\[
A(\q)= \left(1, \ \q^i, \ \ \sqrt{1+|\q|^2} , \ \frac{\q^i}{\sqrt{1+|\q|^2}},\  \frac{\q^i \q ^j}{\sqrt{1+|\q|^2}} \right).
\]
These convergences combined with (\ref{MassEuler}), (\ref{energyEuler}), (\ref{momentumEuler}) and the computations of
Lemma \ref{lema3} lead to
\begin{eqnarray}
&&  \label{Euler1}
\partial_t \left(n_{f_{Eu}} \sqrt{1+|\u_{f_{Eu}}|^2} \right) +   \div_x
(n_{f_{Eu}} \u_{f_{Eu}})= 0,
\\ \nonumber
&& \partial_t \left(-p_{f_{Eu}} \!+ n_{f_{Eu}} \chi(\beta_{f_{Eu}}) (1+|\u_{f_{Eu}}|^2) \right) \\
\label{Euler2}
&& \hspace{2 cm}+\div_x \!\left(\!n_{f_{Eu}} \chi(\beta_{f_{Eu}})\u_{f_{Eu}} \!\sqrt{1+|\u_{f_{Eu}}|^2} \right) = 0,\qquad
\\
\nonumber
&& \partial_t \left(n_{f_{Eu}} \chi(\beta_{f_{Eu}}) \u_{f_{Eu}}^i \sqrt{1+|\u_{f_{Eu}}|^2} \right)\\
\label{Euler3}
&&\hspace{2 cm} + \partial_{x_j} \left(p_{f_{Eu}} \delta^{ij}+ n_{f_{Eu}}\chi(\beta_{f_{Eu}}) \u_{f_{Eu}}^i \u_{f_{Eu}}^j \right)=0,
\end{eqnarray}
respectively, where  $\chi(\beta)= \frac{1}{ \beta}+ \Psi(\beta)$ with $\Psi (\beta)$ given by \eqref{defpsi} in the Appendix. The result that we have sketched can be stated as follows:
\begin{theorem}
Let $f_{Eu}^{\epsilon}(t, x, \q)$  be a sequence
of solutions of the equation \eqref{BGKscalEuler} with nonnegative  initial condition $0\leq f_{Eu}^{\epsilon}(0, x, \q)= f_{Eu}^{\epsilon, I}(x, \q)$ verifying  that,
\[
\int_{\R^3}\int_{\R^3} \big(1+\sqrt{1+|\q|^2}+|x|
+|\ln (f_{Eu}^{\epsilon, I}(x, \q) )|\big)
\,  f_{Eu}^{\epsilon, I}(x, \q) d\q dx <C <\infty .
\]
Assume also that $f_{Eu}^{\epsilon}(t, x, \q)$ converges almost everywhere to 
$f_{Eu}(t, x, \q)$, as $\eps$ goes to
zero. Then, the pointwise limit $f_{Eu}(t,x,\q)$ is a J\"uttner
distribution whose moments are the
corresponding  limits of those of $f_{Eu}^{\epsilon}$. In particular, the
  functions $n_{f_{Eu}}, \u_{f_{Eu}}$ and $\beta_{f_{Eu}}$ associated with $f_{Eu}$ solve the
relativistic Euler equations \eqref{Euler1}--\eqref{Euler3}.
Moreover,  the limit $f_{Eu}(t,x,\q)$ does also satisfy the following entropy inequality
\begin{eqnarray}
&& \partial_t \left(n_{f_{Eu}} \sqrt{1+|\u_{f_{Eu}}|^2}\left(  \ln \Big( \frac{n_{f_{Eu}}}{
M(\beta_{f_{Eu}})}\Big) -\beta_{f_{Eu}} \Psi(\beta_{f_{Eu}}) \right) \right) \nonumber \\
&& \hspace{1.8cm}+ \div_x \left(  n_{f_{Eu}} \u_{f_{Eu}} \left( \ln
\Big( \frac{n_{f_{Eu}}}{ M(\beta_{f_{Eu}})}\Big)-\beta_{f_{Eu}} \Psi(\beta_{f_{Eu}}) \right) \right) \leq 0.
\label{cinqseis}
  \end{eqnarray}
  \end{theorem}
\begin{proof}
We first prove that the solutions are also nonnegative functions. To do that it suffices to  observe that the J\"uttner function on the rhs of  \eqref{BGKscalEuler} is nonnegative and then
$$
\frac{d}{dt}\Big[f_{Eu}^{\epsilon} (t,x+\hat{q}t,\q) \Big] \ge - f_{Eu}^{\epsilon}(t,x+\hat{q}t,\q)
$$
and so, integrating, we deduce  $ f_{Eu}^{\epsilon} (t,x,\q) \ge e^{-t} f_{Eu}^{\epsilon, I}(x-\hat{q}t,\q)\ge 0$. Second, as in the classical kinetic theory, we will prove the following \textit{a priori} estimates
\[
\int_{\R^3}\int_{\R^3} \big(1+\sqrt{1+|\q|^2}+|x|
+|\ln (f_{Eu}^{\epsilon}(t, x, \q) )|\big)
\,  f_{Eu}^{\epsilon, I}(x, \q) d\q dx <C(K) <\infty ,
\]
valid for any $t$ in an arbitrary compact set $[0,K]$. The first two estimates on
\[
\int_{\R^3}\int_{\R^3}f_{Eu}^{\epsilon}(t, x, \q) d\q dx \quad \mbox{and} \quad \int_{\R^3}\int_{\R^3}
\sqrt{1+|\q|^2}f_{Eu}^{\epsilon}(t, x, \q) d\q dx
\]
follow directly, using the non-negativeness, from \eqref{MassEuler} and \eqref{energyEuler} respectively after integration with respect to $x$. Also, multiplying \eqref{MassEuler} by $|x|$ and integrating, we obtain
\[
\frac{d}{dt} \int_{\R^3} \int_{\R^3} |x| f_{Eu}^{\epsilon} \ dxd\q =   \int_{\R^3}\int_{\R^3}  \frac{x \cdot \q}{|x| \sqrt{1+|\q|^2}} f_{Eu}^{\epsilon} \ dx\, d\q\leq C.
\]
Then,
\[
 \int_{\R^3} \int_{\R^3} |x| f_{Eu}^{\epsilon}(t,x,\q) \ dx\, d\q \leq \int_{\R^3} \int_{\R^3} |x| f_{Eu}^{\epsilon, I}(x,\q) \ dx\, d\q + K C, \ \ \forall \ t\in [0,K].
\]
On the other hand, \eqref{entropyscaled} reads now
\begin{equation}
\label{entropy21Euler}
\partial_t \int_{\R^3} f_{Eu}^{\epsilon} \ln(f_{Eu}^{\epsilon})\
d\q + \div_x \int_{\R^3} \frac{\q}{\sqrt{1+|\q|^2}} f_{Eu}^{\epsilon} \ln(f_{Eu}^{\epsilon}) d\q \leq 0,
 \end{equation}
and so, integration on $x\in \R^3$ gives us
\[
\int_{\R^3} \int_{\R^3} f_{Eu}^{\epsilon} \ln(f_{Eu}^{\epsilon}) \ dx\, d\q \leq \int_{\R^3} \int_{\R^3}f_{Eu}^{\epsilon, I} \ln(f_{Eu}^{\epsilon, I}) \ dx\, d\q \ \forall \ t\in [0,K].
\]
The last \textit{a priori} estimate on $f_{Eu}^{\epsilon}  |\log (f_{Eu}^{\epsilon} )|$ follows straightforwardly  from the following trick due to Carleman:
\[
s\, |\ln(s)| \leq s\, \ln(s) +k \, s +2 e^{-\frac{k}{4}}, \qquad \forall s>0, \ \forall k\geq 0.
\]
by choosing $s=f_{Eu}^{\epsilon}(t,x,\q)$, $k=|x|+\sqrt{1+|\q|^2}$ and using previous estimates.

These a priori estimates allow us, via Dunford--Pettis theorem, to pass to the limit weakly in $L^1(\R^3 \times \R^3 )$, uniformly in $t\in [0,K]$. The pointwise convergence together with this weak convergence allow to pass to the limit strongly in $L^1(\R^3 \times \R^3 )$, uniformly in $t\in [0,K]$. Then, the
 associated quantities $n_{f_{Eu}^{\epsilon}}$, $\u_{f_{Eu}^{\epsilon}}$ and $\beta_{f_{Eu}^{\epsilon}}$ pass to the limit strongly to their respective ones. As $J_{f_{Eu}}$ is a continuous function on the variables $n$, $\u$ and $\beta$ we get that $J_{f_{Eu}^{\epsilon}} \to J_{f_{Eu}}$. On the other hand, as we said before, multiplication of \eqref{BGKscalEuler} by $\epsilon$ and  passage to the limit for $\epsilon \to 0$ give that the limits of $J_{f_{Eu}^{\epsilon}}$ and that of $f_{Eu}^{\epsilon}$ coincide, at least in a distributional sense,  then $J_{f_{Eu}} =f_{Eu}$ a.e.

Finishing  the proof only requires to take limits in (\ref{MassEuler})--(\ref{momentumEuler}) to obtain \eqref{Euler1}--\eqref{Euler3} and in  \eqref{entropy21Euler} to obtain \eqref{cinqseis}. The only two involved quantities that pass only weakly to the limit are
\[
\int_{\R} \q f_{Eu}^{\epsilon} (t,x,\q) d\q \qquad \textit{and} \qquad  \int_{\R} f_{Eu}^{\epsilon} (t,x,\q) \ln (f_{Eu}^{\epsilon} (t,x,\q)) d\q,
\]
but this is enough for our purposes, because the pointwise convergence allows to identify their weak limits as
\[
\int_{\R} \q J_{f_{Eu}} (t,x,\q) d\q \qquad \textit{and} \qquad  \int_{\R} J_{f_{Eu}} (t,x,\q) \ln (J_{f_{Eu}} (t,x,\q)) d\q,
\]
respectively. So, computations of  Lemma \ref{sieben} concludes the proof.
\end{proof}

\begin{Remark}
As in the classical case, the connection between kinetic and
ma\-cros\-co\-pic Euler fluid dynamics results from the properties of the
collision operator. These are the conservation equations and the
entropy relation (which implies that the equilibrium is a Maxwellian
distribution for the zeroth-order limit)
\cite{BGL,BGLI,BGLII,Bellouquid2,Saint}. Analogously in the
relativistic case \cite{Kunik2004,speck-strain} these properties
allow to obtain the relativistic Euler equations as an hydrodynamic
limit. Therefore it is very natural to obtain here the same
relativistic Euler equations.
\end{Remark}

\begin{Remark}
Let us mention that the equation of state that we would obtain in the formal hydrodynamical limit to the relativistic Euler equations --a state equation which consists on extrapolating the relations among the thermodynamical quantities of local J\"uttner equilibria to arbitrary solutions of the fluid equations-- is already discussed in \cite{speck-strain}, under the name of kinetic equation of state. The results in \cite{speck-strain,Calvo2012} show that under such equation of state the relativistic Euler system is hyperbolic and causal, and the speed of sound is bounded by $c/\sqrt{3}$.
\end{Remark}

\subsection{The ultra-relativistic Limit of the relativistic BGK model}
\label{UltrarelEuler}
Let us now consider the ultra-relativistic limit of our relativistic BGK model. Some of the results in this sense --as the limiting behavior of the J\"uttner equilibrium or the ultra-relativistic Euler equations-- have been partially discussed elsewhere (\cite{CercignaniKremer,Kunik2001,Kunik2004}), but we include them to give a completely coherent picture. The ultra-relativistic limit corresponds to the third
scaling given in \eqref{3escalas}, that is $m c < < \mu$, and can be motivated as an asymptotic behavior when the rest mass $m$ is very small compared with $\mu /c$, as it happens with neutrinos or even photons in the limit case of zero rest mass. Then, normalizing the remaining two constants  $\bar \omega := \omega  \tau$ and $\bar \eta :=\eta \mu^3/{\cal N}$ as in section \ref{ClassicalLimit}, and defining now $\epsilon := mc/\mu < < 1$, the system \eqref{barBGK} becomes
\begin{equation}
\partial_t f_{ur} +\frac{\bar \q \cdot \nabla_x f_{ur}}{\sqrt{\epsilon^2+|\bar \q|^2}} = \frac{\bar \omega}{\sqrt{\epsilon^2+|\bar \q|^2}} \Big(\frac{\mu^3}{{\cal N}}J_f-f_{ur}\Big),
\label{BGKscaledUR}
\end{equation}
Here $f_{ur}=\mu^3 f /{\cal N}$ denotes the scaled distribution under this ultra-relativistic scaling. As in subsection \ref{ClassicalLimit}, it essentially remains to work on the expression of the right hand side.
To do that, we define  the ultra-relativistic quantities associated to the distribution $f_{ur}$:
\begin{eqnarray}
n_{ur}&:=&  \left[ \left( \int_{\R^3} f_{ur}(\bar t,\bar x,\bar  \q) d \bar \q \right)^2 -  \sum_{i=1}^3  \left( \int_{\R^3} \frac{\bar \q_i}{|\bar \q |} f_{ur}(\bar t,\bar x,\bar \q) d \bar \q \right)^2 \right]^{\frac12}, \nonumber \\
 \u_{ur} &:=&  \left(\int_{\R^3} \frac{\bar \q}{|\bar \q |} f_{ur}(\bar t,\bar x,\bar \q) d \bar \q \right) / n_{ur}, \nonumber \\
\beta_{ur} &:=& 2 \left(\int_{\R^3} \frac{1}{|\bar \q |} f_{ur}(\bar t,\bar x,\bar \q) d \bar \q \right) / n_{ur}. \label{defbetaur}
\end{eqnarray}
We argue analogously as in section \ref{ClassicalLimit} to obtain the following asymptotic expansions
\begin{eqnarray}
\frac{1}{q^0}&:=& \frac{1}{\mu c} \frac{1}{\sqrt{\epsilon^2+|\bar \q|^2}} =  \frac{1}{\mu c}  \frac{1}{|\bar q|} + {\cal O}(\epsilon^2) , \label{q0ur} \\
n_{f} &:=&{\cal N}   \left[ \left( \int_{\R^3} f_{ur}d \bar \q \right)^2\! \! -  \sum_{i=1}^3  \left( \int_{\R^3} \frac{\bar \q_i f_{ur} d \bar \q }{\sqrt{\epsilon^2+|\bar \q|^2}}  \right)^2 \right]^{\frac12} \!\! = {\cal N} n_{ur} +  {\cal O}(\epsilon^2), \label{nur} \\
  \u_f  &:=& \frac{{\cal N}}{n_{f}}  \int_{\R^3} \frac{\bar \q  }{\sqrt{\epsilon^2+|\bar \q|^2}} f_{ur}(\bar t,\bar x,\bar  \q) d \bar \q =  \u_{ur} +  {\cal O}(\epsilon^2) .\label{uur}
\end{eqnarray}
Also, to obtain the behavior of $\beta_f$ we first use \eqref{betadef} combined with the expansions \eqref{q0ur} and \eqref{nur} and definition \eqref{defbetaur} to deduce
\begin{eqnarray*}
\displaystyle\frac{K_1}{K_2}(\beta_f) &=&
\frac{\epsilon {\cal N}}{n_f} \,   \displaystyle \int_{\R^3}  \frac{f_{ur} }{\sqrt{\epsilon^2+|\bar \q|^2}} d\q = \frac{\epsilon}{n_{ur}}\,   \displaystyle \int_{\R^3} \frac{1}{|\bar \q |} f_{ur}(\bar t,\bar x,\bar \q) d \bar \q  +   {\cal O}(\epsilon^2)\\
& =& \frac{\epsilon}{2} \beta_{ur} +   {\cal O}(\epsilon^2).
\end{eqnarray*}
In particular,  note that now $K_1(\beta_f)/K_2 (\beta_f)$ approaches 0, so equivalently $\beta_f$ is small and justifies the  following expansion, which is valid for $\beta_f < <1$  (see \cite{librotablas}),
\[
\displaystyle\frac{K_1}{K_2}(\beta_f) = \frac{\beta_{f}}{2} + {\cal O}(\beta_f^2).
\]
Using finally  these two expansions of $K_1/K_2$ we conclude
\begin{equation}
\beta_f= \epsilon \,  \beta_{ur} + {\cal O}(\beta_f^2) + {\cal O}(\epsilon^2).
\label{betaur}
 \end{equation}
We are now ready to analyze the behavior of $\frac{\mu^3}{{\cal N}}J_f$. We first write  it as
\[
\frac{\mu^3}{{\cal N}}J_f = \frac{ n_f }{{\cal N} \, \epsilon^3 \, M(\beta_f)} \exp
\Big\{- \frac{\beta_f}{\epsilon}  \Big( \sqrt{ 1+ |\u_f|^2} \sqrt{\epsilon^2+ |\bar \q |^2} -\u_f\cdot \bar  \q \Big) \Big\}
\]
and then, using  the asymptotic relation (see \cite{librotablas})
$$
\frac{1}{M(\beta_f)} = \frac{1}{8 \pi}\beta_f^3 + {\cal O}(\beta_f^5) \quad \mbox{for}\ \beta_f < <1
$$
and the expansions  \eqref{nur}, \eqref{uur} and \eqref{betaur}, we deduce that
\[
\frac{\mu^3}{{\cal N}}J_f = \frac{n_{ur} \beta_{ur}^3 }{8\pi } \exp
\Big\{- \beta_{ur} \Big(  |\bar \q |\sqrt{ 1+ |\u_{ur}|^2}   -\u_{ur}\cdot \bar \q \Big) \Big\} + {\cal O}(\epsilon).
\]
Then, up to order $\epsilon$,  \eqref{BGKscaledUR} becomes the ultra-relativistic BGK equation
\begin{eqnarray*}
&&\hspace{-1cm} \partial_t f_{ur} +\frac{\bar \q}{|\bar \q|}  \cdot \nabla_x f_{ur}  \\
&&= \frac{\bar \omega}{|\bar \q|} \left(
\frac{n_{ur} \beta_{ur}^3 }{8\pi } \exp
\Big\{- \beta_{ur} \Big( |\bar \q | \sqrt{ 1+ |\u_{ur}|^2}   -\u_{ur} \cdot \bar \q \Big) \Big\}
-f_{ur}\right).
\end{eqnarray*}
To conclude we compute various macroscopic quantities of local equilibria in this limit.
The calculations are based on the asymptotic expansion of  $\Psi(\beta)$ \eqref{defpsi}
\begin{equation}
\label{psinear}
\Psi(\beta) \sim \frac{3}{\beta} + {\cal O}( \beta) \quad \mbox{for}\ \beta<<1
\end{equation}
(see \cite{librotablas} for instance). Thanks to Lemma \ref{average_energy} we have
$$
e_{J_f}=c^2n_f\Psi(\beta_f) \sim \frac{3}{\beta_{ur}} c^2 n_{ur} \frac{{\cal N}}{\epsilon}:= e_{ur}\frac{{\cal N}}{\epsilon}.
$$
In a similar way,
$$
p_{J_f}=c^2\frac{n_f}{\beta_f} \sim c^2 \frac{n_{ur}}{\beta_{ur}} \frac{{\cal N}}{\epsilon}:=p_{ur}\frac{{\cal N}}{\epsilon}.
$$
Note that in this regime the energy amounts to three times the pressure. We can also use Lemma \ref{lema4} to obtain
$$
 N_{J_f}^\mu \sim {\cal N} n_{ur} \ u_{ur}^\mu,\quad T^{\mu \nu} \sim - \frac{{\cal N}}{\epsilon} p_{ur} g^{\mu \nu}+ 4 \frac{{\cal N}}{\epsilon} p_{ur}\, u_{ur}^\mu\ u_{ur}^\nu.
$$

\section{Linearization of the relativistic BGK operator}
\label{displaylinearization}
From now on, we will use the dimensionless system \eqref{barBGK} by using the choice \eqref{normalscale}. If we skip the superscript $``\bar{\ }"$ to simplify the notation,  it becomes
\begin{equation}
\label{BGKhtheorem}
\partial_t f+\hat{q}\cdot\nabla_x f= \frac{1}{q^0} (J_f - f),
\end{equation}
where $q^0=\sqrt{1+|\q|^2}$, $\hat{q}=\q/q^0$ and the local equilibria  given by J\"uttner distributions  (or relativistic Maxwellian) takes the dimensionless form:
$$
J (n,\beta ,\u;\q) = \frac{n}{M(\beta)} \exp \{- \beta (\sqrt{(1+|\u|^2)(1 + |\q|^2)} - \u \cdot \q)Ê\}.
$$
We will study solutions close to a family of global equilibria, those of the form
$$
  J^0\equiv J(1,\beta_0,0;\q)=\frac{e^{-\beta_0 \sqrt{1+|\q|^2}}}{M(\beta_0)}=\frac{e^{-\beta_0 q^0
}}{M(\beta_0)}, \qquad 0< \beta_0 < \infty.
$$
From the mathematical point of view it is handy to bring in a new variable $\alpha$ defined by means of
\begin{equation}
\beta = \equis(\alpha):= \left(\frac{K_1}{K_2}\right)^{-1}(\alpha),
\quad \mbox{where} \quad \alpha:= \frac{1}{n_f}\int_{\R^3} f \frac{d\q}{q^0}.
\nonumber
\end{equation}
It is in one-to-one correspondence with $\beta$, see Section \ref{mono}. Many related formulae can be written down in a more
compact way using this parameter instead. Thus, we hereafter consider
$$
J^0 \equiv J(1,\alpha_0,0;\q)=\frac{e^{-\equis(\alpha_0) \sqrt{1+|\q|^2}}}{M(\equis(\alpha_0))}=\frac{e^{-\equis(\alpha_0) q^0
}}{M(\equis(\alpha_0))}, \quad 0< \alpha_0= \equis^{-1}(\beta_0)< 1.
$$

Our aim will be to find an equation for $f$ such that $g=J^0 + \sqrt{J^0}f$
is a solution to the relativistic BGK model  (\ref{BGKhtheorem}). To do that we will linearize the relativistic BGK operator around $J^0$
by using Taylor's formula. It will be useful to control the moments of the perturbation.

\begin{Remark}
\label{negativo}
In the following we shall consider perturbations $h=\sqrt{J^0}f$ which need not be positive, so that objects like $n_h$ become meaningless.
Nevertheless, we notice that in the following the object $n_h$ does only appear involved
in combinations like $n_h^2$, $n_hu_h$, $n_h\sqrt{1+|u_h|^2}$ or $n_h\alpha_h$. These are just integrals which make perfect sense even if $h
$ takes negative values.
 So, in order not to distinguish the cases in which $h\ge 0$ from those in which this is not so, we \textnormal{define}
 \begin{eqnarray*}
 n_h^2 &:= &\left(\int_{\R^3} h \ d\q \right)^2 - \sum_{i=1}^3 \left(\int_{\R^3} q^i h \ \frac{d\q}{q^0
} \right)^2,
 \\
  n_h \u_h &:= &\int_{\R^3} \q h \ \frac{d\q}{q^0
},
  \\
   n_h \sqrt{1+|\u_h|^2} &:= & \int_{\R^3} h \ d\q,
   \\
    n_h \alpha_h &:= &\int_{\R^3}  h \ \frac{d\q}{q^0
},
 \end{eqnarray*}
and in this way we can keep the same notation as in the nonnegative case. Let us also remark that $n_h\u_h,\, n_h \sqrt{1+|\u_h|^2}$ and $n_h \alpha_h$ are linear in $h$.
\end{Remark}

\begin{lemma}
\label{le10}
Let $g=J^0 + h$ with $h= f \sqrt{J^0} $. Then, the associated
macroscopic quantities to the function $g$ are given  by the following formulae:
\begin{eqnarray*}
&&(i) \  n_g = \sqrt{\left(1+ n_h \sqrt{1+|\u_h|^2} \right)^2 - n_h^2 |\u_h|^2} = \sqrt{1+ n_h^2 +2n_h \sqrt{1+|\u_h|^2}} ,
\\
&& (ii) \  n_g \u_g = n_h \u_h, \\
&& (iii) \ n_g\alpha_g  = \int_{\R^3} g \frac{d\q}{q^0
}   =n_h \alpha_h + \alpha_0.
\end{eqnarray*}
\end{lemma}
\begin{proof}
It is mostly straightforward. We first note that
$$
n_g \u_g=  \int_{\R^3} \q (J^0+h ) \frac{d\q}{q^0
} = \int_{\R^3} \q \,
h \frac{d\q}{q^0
} =n_h \u_h,
$$
which proves $(ii)$. Now we obtain
$$
n_g^2( 1 + |\u_g|^2) = \left( \int_{\R^3} (J^0+h) \ d\q \right)
^2\! \! = \left(1+\int_{\R^3} h\ d\q \right)^2\!\! = \left(1+ n_h \sqrt{1+|\u_h|^2}\right)^2\! \! ,
$$
which combined with $(ii)$ gives $(i)$. To deal with $(iii)$ it suffices to write
$$
n_g \alpha_g = \int_{\R^3} (J^0+h)\frac{d\q}{q^0
}= \alpha_0 + n_h \alpha_h.
$$ This completes the proof.
\end{proof}

Let us denote by $D^2_{(n, \u, \alpha)} J^\theta $ the $5\times 5$ Hessian matrix of the J\"uttner function with respect to the
variables
$(n, \u, \alpha)$, at the point
$(n_\theta, \u_\theta, \alpha_\theta): =
 \theta(n_g,\u_g,\alpha_g)+(1-\theta)(1,0,\alpha_0)$, for any given $\theta \in [0,1]$. Being given $J_g=J(n_g,\alpha_g,\u_g; \q)$, one has
the following Taylor expansion.

\begin{lemma}
Let
$g=J^0 + h$ be a solution of the relativistic BGK equation \eqref{BGKhtheorem} with $h= f \sqrt{J^0} $. Then the following equality
\begin{eqnarray*}
J_g - J^0 &=&  \left(n_g-1\right)J^0 (\q) + \frac{n_h\,  \beta_0}{n_g}  \u_h\cdot \q \, J^0 (\q)
\\
&& -\left( \frac{n_h \alpha_h + \alpha_0}{n_g} - \alpha_0 \right) \equis'(\alpha_0)
\left(\frac{M' (\beta_0)}{M (\beta_0)} +q^0
\right) J^0 (\q)
\\
&& + \int_0^1 (1-\theta) (n_g -1, \u_g, \alpha_g -\alpha_0) D^2 J^\theta (n_g -1, \u_g, \alpha_g -\alpha_0) d\theta\\
&&:= T_1+T_2+T_3+\tilde{\Gamma}
\end{eqnarray*}
holds.
\label{LemmaTaylor}
\end{lemma}

\begin{proof}
The proof is a direct application of Taylor formula together with previous Lemma \ref{le10}, after computations of the
first derivatives of the J\"uttner function with respect to $n$, $\u$  and $\alpha$, which read
$$
\frac{\partial J}{\partial n}= \frac{J}{n}, \qquad \nabla_\u J = \beta_0 \left(\q -
\frac{q^0
}{\sqrt{1+|\u|^2}}\u \right)J,
$$
$$
\frac{\partial J}{\partial \alpha} = -\equis'(\alpha)\left(\frac{M'(\equis(\alpha))}{M(\equis(\alpha))} +
q^0
\sqrt{(1+|\u|^2)} -\u\cdot \q \right)J
$$
respectively. Note that we only need to use the above formulae for the values $(n,
\alpha, \u)$= $(1,\alpha_0,0)$, which simplifies the expressions.
\end{proof}

For the next step we set
$$
\kappa_0
:= \frac{3
\alpha_0}{\beta_0}+\alpha_0^2 -1.$$
Recall that $g = J^0 + f \sqrt{J^0}=J^0+h$. Now we let $P(f)$ be the following linear
operator (see Remark \ref{negativo})
\begin{eqnarray}
P(f) &:=&\left[ \left(\frac{\alpha_0q^0
-1}{
\kappa_0
}\right) n_h\sqrt{1 +
|\u_h|^2}
+ \left( \frac{\Psi(\beta_0)-q^0
}{
\kappa_0
}\right) n_h
\alpha_h + \beta_0  n_h \u_h \cdot \q \right] \sqrt{J^0} \nonumber
\\
&&\hspace{-1 cm}=\left(\frac{n_h(\Psi(\beta_0) \alpha_h -  \sqrt{1+|\u_h|^2})}{
\kappa_0
}\right) \sqrt{J^0} +
\left(\frac{n_h( \alpha_0 \sqrt{1+|\u_h|^2} - \alpha_h)}{
\kappa_0
}\right) q^0
 \sqrt{J^0} \nonumber
 \\
&&+ \beta_0  n_h \u_h \cdot \q  \sqrt{J^0} := A_f \sqrt{J^0} +B_f q^0
 \sqrt{J^0} + C_f \cdot \q \sqrt{J^0}.
\label{operatorP}
\end{eqnarray}
The following result shows how does the operator $P$ enter in our framework.
\begin{lemma} Let $g=J^0 + h$ be a solution of the relativistic BGK equation \eqref{BGKhtheorem} with
$h= f \sqrt{J^0} $. Then $f$ verifies
$$
\partial_t f + \hat{q} \cdot \nabla_x f = \frac{1}{q^0
}(P(f)-f)+ \Gamma(f)
$$
where the  nonlinear term $\Gamma(f)$ is  given by
\begin{eqnarray*}
\Gamma (f)&=& \left( \frac{n_h^2}{2} -  \frac{\#^2}{4(1+\# /2 +
\sqrt{1+\#})} \right) \frac{\sqrt{J^0}}{q^0
}
\\
&&
+\beta_0  \frac{\# }{(1+\sqrt{1+\#})\sqrt{1+\#} }\, n_h \u_h\cdot \q \frac{\sqrt{J^0}}{q^0
}\\
&& + n_h \alpha_h\,
  \frac{\#}{(1+\sqrt{1+\#})\sqrt{1+\#} }\frac{\sqrt{J^0}}{q^0
}\\
 && - \alpha_0 \left(\frac{n_h^2}{2} + \frac{\#^3 -3\#^2}{2(2+\#-\#^2+2\sqrt{1+\#})}\right)\frac{\Psi(\beta_0)-q^0
}{
\kappa_0
}\frac{\sqrt{J^0}}{q^0
}\\
&& + \frac{1}{q^0
 \sqrt{ J^0}} \int_0^1 (1-\theta) (n_g -1, \u_g,
\alpha_g -\alpha_0) D^2 J^\theta (n_g -1, \u_g,
\alpha_g -\alpha_0) d\theta.
\end{eqnarray*}
Here $J^\theta:=J(
\theta (n_g, \u_g,
\alpha_g)+(1-\theta)(1,0,\alpha_0))$ and $\#$ is given by $\sqrt{1+\#}=n_g=\sqrt{1+ n_h^2+2 n_h \sqrt{1+|\u_h|^2} }$.
\end{lemma}
\begin{proof}
We first remark that if $g=J^0 + f \sqrt{J^0}$ verifies \eqref{BGKhtheorem}, then $f $ satisfies
$$
\partial_t f + \hat{q} \cdot \nabla_x f = \frac{1}{q^0
} \left( \frac{J_g-J^0}{\sqrt{J^0}} -f \right)
$$
and the proof consists in identifying the linear and the  nonlinear parts on the right hand side. Actually, we have to prove that
\[
 \frac{1}{q^0
} \left(  \frac{J_g-J^0}{\sqrt{J^0}} -f \right)   =  \frac{1}{q^0
} \left( P(f)-f\right) + \Gamma(f).
\]
In order to do that we use Lemma \ref{LemmaTaylor} to decompose $J_g - J^0$ into its linear and nonlinear parts. First, we use the following
equality
\[
\sqrt{1+\#}-1 =\frac{\#}{2}-\frac{\#^2}{4(1+\#/2+\sqrt{1+\#})},
\]
to decompose $T_1$ as
\begin{eqnarray*}
T_1&=& \left(\sqrt{1+\#}-1\right)J^0\\
&=&\left(  \frac{n_h^2 +2n_h \sqrt{1+|\u_h|^2}}{2} -  \frac{\#^2}{4(1+\# /2 + \sqrt{1+\#})}\right) {J^0}\nonumber
\\
&=& n_h \sqrt{1+|\u_h|^2}  J^0 + \left(  \frac{n_h^2}{2} -  \frac{\#^2}{4(1+\# /2 + \sqrt{1+\#})}\right) {J^0}:=L_1+\Gamma_1.
\end{eqnarray*}
Second, by using the equality
\begin{equation}
\frac{1}{\sqrt{1+\#}}=1+\frac{\#}{\sqrt{1+\#}(1+\sqrt{1+\#})}
\label{tmp1}
\end{equation}
we can write $T_2$ as
\begin{equation*}
T_2=\beta_0 n_h \u_h\cdot \q  \, J^0 +
\frac{\beta_0 \# }{\sqrt{1+\#}(1+\sqrt{1+\#})}\, n_h \u_h\cdot \q \, J^0:=L_2+\Gamma_2.
\end{equation*}
For the third term,  we first note that $M'(\beta)/M(\beta)=-\Psi(\beta)$, see \eqref{defpsi} in the Appendix, and also the fact that
$$
\equis'(\alpha)=\frac{1}{(K_1/K_2)'(\beta)}= \frac{1}{\frac{3}{\beta}\frac{K_1}{K_2}(\beta)+\left(\frac{K_1}{K_2}(\beta)\right)^2-1}= \frac{1}{\frac{3
\alpha}{\beta}+\alpha^2 -1},
$$
so that $\equis'(\alpha_0)=1/
\kappa_0
$. Then, $T_3$ can be rewritten as
\[
T_3=\frac{n_h \alpha_h}{\sqrt{1+\#}}
\frac{\Psi(\beta_0)-q^0
}{
\kappa_0
} J^0 +\alpha_0
\left( \frac{1}{\sqrt{1+\#}} -1 \right)
\frac{\Psi(\beta_0)-q^0
}{
\kappa_0
} J^0:=T_{31}+T_{32}.
\]
Using again (\ref{tmp1}), we can write $T_{31}$ as
\[
T_{31}=
\frac{\Psi(\beta_0)-q^0
}{
\kappa_0
}\,  n_h \alpha_h \, J^0+
\frac{\# }{\sqrt{1+\#}(1+\sqrt{1+\#})}\, n_h \alpha_h \, J^0:=L_3+\Gamma_3.
\]
Finally, using the equality
\[
\frac{1}{\sqrt{1+\#}}-1=-\frac{\#}{2}-\frac{\#^3-3\#^2}{2(2+\#-\#^2+2\sqrt{1+\#})},
\]
we can write $T_{32}$ as
\begin{eqnarray*}
T_{32}&=&\alpha_0 \Bigg(- \frac{n_h^2 +2n_h \sqrt{1+|\u_h|^2}}{2}
\\&& \hspace{0.8cm}-\frac{\#^3-3\#^2}{2(2+\#-\#^2+2\sqrt{1+\#})} \Bigg)
\frac{\Psi(\beta_0)-q^0
}{
\kappa_0
} J^0 \nonumber
\\
&=& - \alpha_0 \, \frac{\Psi(\beta_0)-q^0
}{
\kappa_0
} \, n_h \sqrt{1+|\u_h|^2} \, J^0 \nonumber \\
&& -\alpha_0 \left( \frac{n_h^2 }{2}+\frac{\#^3-3\#^2}{2(2+\#-\#^2+2\sqrt{1+\#})} \right)
\frac{\Psi(\beta_0)-q^0
}{
\kappa_0
} J^0:= L_4+\Gamma_4.
\end{eqnarray*}
We only have to remark that $ P(f)=(L_1+L_2+L_3+L_4)/\sqrt{J^0}$ and that $ \Gamma (f)= (\Gamma_1+\Gamma_2+\Gamma_3+\Gamma_4+
\tilde{\Gamma})/(q^0
 \sqrt{J^0})$. Just note that
$$
1 -\alpha_0 \frac{\Psi(\beta_0)-q^0
}{
\kappa_0
}=\frac{\alpha_0 q^0
-1}{
\kappa_0
}
$$
helps to deal with the factor of $u_h\sqrt{1+|\u_h|^2}$ in the sum of $L_1$ and $L_4$.
\end{proof}


\subsection{Analysis of the linearized operator: the projector over the distinguished space}
In this paragraph we determine some properties of operator $P$ that will be needed for future analysis.
We could consider the interplay of the projector with different scalar products. At least two of them come quickly into mind:
$$
 \langle f,\ g\rangle = \int_{\R^3} f(\q) \ g(\q) \ d\q \quad \mbox{and}\quad
 \langle f,\ g\rangle_{q_0} = \int_{\R^3} \frac{f(\q) \ g(\q) }{q_0}\ d\q.
$$
To proceed, let $N$ be the five dimensional space given by
$$N= span\{ \sqrt{J^0}, q^\mu \sqrt{J^0} \}.$$

\begin{lemma} The linear operator $P$ given in (\ref{operatorP}) is the projection from $L^2(\R^3)$ onto $N$ with respect to the scalar product $\langle \cdot ,\ \cdot\rangle_{q_0}$. In particular, it is self-adjoint with respect to that scalar product.
\label{lemaP11}
\end{lemma}
\begin{proof}
We first observe, from the definition  (\ref{operatorP}) of  $P(f)$, that it is a linear combination of the basis in $N$ and that
\begin{eqnarray*}
P( \sqrt{J^0})= \sqrt{J^0},\quad  P( q^i \sqrt{J^0})=  q^i \sqrt{J^0}, \quad P( q_0 \sqrt{J^0})=q_0 \sqrt{J^0},
\end{eqnarray*}
which is true by using a direct computation and the fact that
$1-\frac{\alpha_0}{
\kappa_0
}\Psi(\beta_0) =-\frac{1}{
\kappa_0
}$. Then $P$ acts on $N$ as the identity.

On the other hand, we can easily compute, for any $f$ and $g$, that
\begin{eqnarray*}
&&
\hspace{-0,5cm}\langle P( f),g \rangle_{q_0} \\&&= \int \!\!\! \int \!
\left(  \frac{\alpha_0}{
\kappa_0
}-  \frac{p_0+q_0}{
\kappa_0
 p_0q_0 } + \frac{\Psi(\beta_0)}{
\kappa_0
 p_0 q_0}  +\beta_0 \frac{\p\cdot \q}{p_0 q_0}\right)  f(\q) g(\p) \sqrt{J^0 (\q)}  \sqrt{J^0(\p)} d\q d\p
\end{eqnarray*}
which is obviously a symmetric expression and shows that $P$ is self-adjoint.

Finally, we write for any $f$ and any $g\in N$
\[
\langle f-P(f),g \rangle_{q_0} =\langle f,g \rangle_{q_0}-\langle P(f),g \rangle_{q_0}=\langle f,P(g) \rangle_{q_0}-\langle f,P(g) \rangle_{q_0}=0,
\]
which concludes the proof. \end{proof}

%
Let us define  $Lf=\frac{P(f)-f}{q^0}$ as the linearized relativistic BGK operator. We have the following result.
\begin{lemma}
\label{linea}
 The operator $L$ satisfies the following properties:
 \begin{enumerate}

 \item It is self-adjoint with respect to the scalar product $\langle \cdot ,\ \cdot \rangle$.

 \item $Ker (L)=N$.

 \item $L$ is non positive. In fact,
  \begin{equation}
  \nonumber
  \langle Lf,f \rangle =- \left\langle (I-P)f, (I-P)f \right\rangle_{q^0
} \leq 0.
 \end{equation}

\item
$L$ is  decomposed as $ L= -1/q^0 Id+K$, where $K$ is the
compact operator in $L^2(\R^3)$ given by $K=\frac{P}{q^0}$.
\end{enumerate}
\end{lemma}
\begin{proof}
By definition, $Lf=0$  is equivalent to $f=P(f)$, so clearly $f\in R(P)=N$. The self-adjointness of the linear operator $L$ is
 easily obtained computing from the definition (\ref{operatorP}) that
 \begin{eqnarray*}
 \langle Lf,g \rangle&=&\int_{\R^3} \int_{\R^3}
 \left(  \frac{\alpha_0}{
\kappa_0
}-  \frac{1}{
\kappa_0
 p_0} - \frac{1}{
\kappa_0
 q^0
} + \frac{\Psi(\beta_0)}{
\kappa_0
 p_0 q^0
}  +\beta_0 \frac{\p\cdot \q}{p_0 q^0
}\right)
\\
 &&\qquad \times f(\q) g(\p) \sqrt{J^0 (\q)}  \sqrt{J^0(\q)} d\q d\p
-\int_{\R^3} \frac{g (\q) f(\q)}{q^0
} d\q.
\end{eqnarray*}
To prove \emph{3} we first notice that
$$ \langle Lf,f \rangle = -\Big \langle \frac{P(f)-f}{q^0}, P(f)-f
\Big\rangle +\Big\langle \frac{P(f)}{q^0}, P(f) \Big\rangle -\Big\langle
\frac{P(f)}{q^0}, f \Big\rangle.
$$
Therefore, using Lemma \ref{lemaP11} we observe that the last two terms cancel and this completes the proof.

The decomposition of the linear operator $L$ is trivial. Note that the operator $K=P/q^0$ is compact because its range $R(P)=N$ has finite dimension. Moreover, for future development, let us write it in the Hilbert-Schmidt form as follows
$$
K(f)= \int_{\R^3} k(\q,\q_1)f(\q_1)\, d\q_1,
$$
where the kernel $ k$ is given by
\begin{eqnarray*}
k(\q,\q_1)= \frac{\sqrt{J^0}(\q)\sqrt{J^0}(\q_1)}{q_0}\bigg\{ \left(1 - \frac{\alpha_0(\Psi(\beta_0)-q_0)}{\kappa_0}\right)
\hspace{1.8cm}\\+ \left( \frac{\Psi(\beta_0)-q_0}{\kappa_0
\sqrt{1+q_1^2}}\right)+\beta_0 \frac{\q \cdot \q_1}{\sqrt{1+q_1^2}}\bigg\}.
\end{eqnarray*}
Thus, $K$ defines a Hilbert-Schmidt operator in $L^2(\R_{\q}^3)$, as it is easy to see that $k(\q,
\q_1)\in L^2(\R^3\times
\R^3)$.
\end{proof}

\section{Existence of solutions to the linearized Relativistic
BGK equation}
\label{existencelinearization}

In this section we will show that the initial value
problem for the linearized equation
\begin{equation}
\label{eqlineal}
\partial_t f + \hat{q} \cdot \nabla_x f = Lf
\end{equation}
 has a unique weak solution in $L^2(\R_x^3 \times \R_{\q}^3)$ which is global in time.  This is done by means of the semigroup
representation of the solution, analogous to that given by \cite{EP,N1} for the non-relativistic Boltzmann equation.

To proceed we need to introduce some notation.
Let $l\geq 0$ and let $H^l(x)$ denote the Sobolev space of
$L^2(\R_x^3)$ functions, whose derivatives up to order
$l$ belong to $L^2(\R_x^3)$. We shall denote the partial
Fourier transform in
$x$ of $f\in L^2(\R_x^3 \times \R_{\q}^3)$ as follows:
$$
\hat{f}(\zeta,\q)= (2\pi)^{-\frac{3}{2}}\int_{\R^3} \exp(-i
\zeta x) f(x,\q) \, dx.
$$
Then, we set  $\hat{H}^l (\zeta)$ as the image under Fourier transform of the space
$H^{l}(x)$, with the following norm:
$$
\left\| f\right\|_{\hat{H}^l(\zeta)} = \left\| (1+\mid \zeta
\mid)^{\frac{l}{2}} \hat{f}(\zeta)\right\|_{L^2(\zeta)}=
\left\| f(x)\right\|_{H^l(x)}.
$$
 Finally, let $H_l$ be the Hilbert space $L^2(\R_{\q}^3, H^l(x))$ with the
norm given by
$$
 \left\| f\right\|_l = \sqrt{\int_{\R^3} \left\| f(\cdot,\q)\right\|^2_{H^l(x)}d\q}\ .
 $$
We define the operator $B$ acting on $H_l$ as
$$B= L- \frac{\q}{q^0
} \nabla_x \ . $$
Being $L$ a perturbation of the compact operator $K$, then its domain is given by
$$
D(B)=\left\{f\in H_l / \frac{\q}{q^0
} \nabla_x f +
\frac{1}{q^0
}f \in H_l \right\}.
$$

Now the equation \eqref{eqlineal} can be
written down as  $\frac{\partial f}{\partial t}= Bf$. Given $f \in H_l$, we can take the Fourier transform with respect to the spatial variable in the above equation, thus obtaining
$$
\frac{\partial \hat{f}}{\partial t}= \hat{B}\hat{f},
$$
where
$$
\hat{B}= L-i \frac{\zeta \cdot \q}{q^0
}.
$$
Following now the standard
techniques for the non-relativistic case \cite{N2,UK}, an important
property of $\hat{B}$ is obtained. Namely,

\begin{theorem} For each $\zeta \in \R^3$, the operator $\hat{B}$
generates a strongly continuous contraction semigroup on
$L^2(\R^3_{\q})$ such that, for any $f\in L^2(\R^3_{\q})$, one has
$$\left\| \exp(t\hat{B}) f(\q)\right\|_{L^2(\R^3_{\q})} \leq \left\|
f(\q)\right\|_{L^2(\R^3_{\q})}, \quad t\geq 0.$$
\end{theorem}

The operator $\hat{B}$ can be used to construct an explicit
representation of the semigroup $\exp(t\hat{B}) $ for $f\in H_l$.
Using the same argument as in the non-relativistic case
\cite{EP,N2}, one can show that the operator $B$
generates a strongly continuous contraction semigroup on $H_l$. It
is also shown that such semigroup is given explicitly as:
$$
\exp(tB)f(x,\q)= (2\pi)^{-\frac{3}{2}}\int_{\R^3} \exp(i
\zeta x)\, \exp(t\hat{B})\hat{f}(\zeta,\q)\ d\zeta .
$$
 Moreover, the
following estimate holds for any $t\geq0$:
$$
\left\|
\exp(tB)f(x,\q)\right\|_l \leq \left\| f(x,\q)\right\|_l.
$$
Then we can state finally the following result:
\begin{theorem}
Let $f_0 \in H_l$ with $l\geq 0$. Then, there exists a unique $f(t,x,\q)$ global in time solution of (\ref{eqlineal}), satisfying $f(t)\in H_l$ and
$$\left\|
f(t)\right\|_l \leq \left\| f_0\right\|_l, \quad \forall \ t\geq 0.
$$
\end{theorem}
%

\setcounter{section}{0}%
\setcounter{subsection}{0}%
\renewcommand\thesection{\Alph{section}}

\section{Appendix}

 The aim of this Appendix is to compute some of the various quantities involved in this paper in order to make it  easier  to follow. For the sake of simplicity,  it is enough to perform these computations for the dimensionless quantities as in Sections \ref{displaylinearization} and \ref{existencelinearization}, assuming the scaling \eqref{normalscale} or, equivalently, to assume $m=c=\omega=\eta=1$.

\subsection{Lorentz invariance} {\ }

Let $\Lambda$ be a Lorentz boost (i.e. a linear isometry with respect to the Minkowsky metric) in $\R_q^4$. As we are dealing with particles of unit rest mass (the so-called mass shell condition), this transformation $\Lambda$ can be meaningfully seen as acting on $\R_{\q}^3$.
 Given any distribution function $f$, we can define a new distribution function $f_\Lambda$ by means of
$$
  f_\Lambda(t,x,\q) = f(t,x,\Lambda \q).
$$
As
\begin{equation}
   v_\mu z^\mu =( \Lambda v)_\mu (\Lambda z)^\mu \quad \mbox{for any}\ v,\ z\in \R^4,
\label{scalarprod}
\end{equation}
we can check directly that $n, e, p$ and $\sigma$ are Lorentz invariant, i.e. $n_f = n_{f_\Lambda}$ and so on, for any Lorentz boost $\Lambda$. Moreover, making use of the fact that the ratio $\frac{d\q}{
q^0
}$ is invariant under the action of $\Lambda$, we see that
\begin{eqnarray}
\nonumber
  \int_{\R^3} q^\mu f_\Lambda(t,x,\q) \frac{d\q}{
q^0
} &=&  \int_{\R^3} (\Lambda^{-1} q^\mu) f(t,x,\q) \frac{d\q}{ q^0
}
\\
\nonumber
&=& \Lambda^{-1}\left(\int_{\R^3} q^\mu f(t,x,\q) \frac{d\q}{ q^0
}\right).
\end{eqnarray}
This means that
\begin{equation}
\nonumber
 u_{f_\Lambda} = \Lambda^{-1}u_f
\end{equation}
(e.g. macroscopic boosts on the local velocity of the system are uniquely determined by the action of the same boosts --in a contravariant way-- on the microscopic local velocities of the gas). The Lorentz invariance of the volume element $\frac{d\q}{ q^0}$ shows also that $\beta$ is Lorentz invariant, as the ratio
$$
\frac{K_1(\beta)}{K_2(\beta)} = \frac{1}{n_f} \displaystyle \int_{\R^3}  f \frac{d\q}{q^0}
$$
is Lorentz invariant too. We summarize these facts as:
\begin{lemma}
\label{lorentz}
Given any distribution function $f$, the scalar quantities $n_{f_\Lambda}$, $e_{f_\Lambda}$, $p_{f_\Lambda}$, $\sigma_{f_\Lambda}$ and $\beta_{f_\Lambda}$ are Lorentz invariant. The vector $u_f$ transforms according to $u_{f_\Lambda} = \Lambda^{-1}u_f$.
\end{lemma}

It is instructive to consider the special case in which the distribution function induces a local velocity that is found to be zero; that is, the physical objects that we are representing are at rest with respect to the reference frame that we use to describe them. This situation corresponds to distributions $f$ having $u_f = (1,0,0,0)$ --Lorentz rest frame. It is useful to display formulas for the macroscopic quantities of the gas in this case, as the computations are simpler than in the general case and the results can be related to a generic distribution by means of Lorentz boosts.
These read now:
\begin{equation}
\nonumber
n_f= N^0 = \int_{\R^3} f \ d\q,
\end{equation}
\begin{equation}
\label{efacil}
e_f =  T^{00}= \int_{\R^3} \sqrt{1 + |\q|^2} f \ d\q,
\end{equation}
\begin{equation}
\label{pfacil}
p_f = \frac{1}{3} [T^{11}+T^{22}+T^{33}]= \frac{1}{3}\int_{\R^3} |\q|^2  f \ \frac{d\q}{q^0
}.
\end{equation}

\subsection{Computation of the moments of the J\"uttner equilibrium}

We will need a more precise information about the moments of the relativistic Maxwellian. First we list for convenience some of them that can be easily computed in the Lorentz rest frame. Notice that in this case  the J\"uttner equilibrium reduces to
$$
J(n,\beta,0;\q)= \frac{n}{M(\beta)} \exp \{- \beta \sqrt{1 + |\q|^2} \}.
$$
For future reference we point that,
using modified Bessel functions for the non-negative integer number $j$
$$
K_j(\beta) =  \int_0^\infty \cosh(jr)\exp\{-\beta
\cosh(r)\}dr,
$$
we can simplify some of the related formulae. For instance, we
 can write the function $M(\beta)$ given in \eqref{eme} as
\begin{equation}
\label{diecinueve}
M(\beta)=  \frac{4\pi}{ \beta} K_2(\beta).
\end{equation}
To simplify the  notation we will introduce the function $\Psi$ defined as follows
\begin{equation}
\label{defpsi}
 \Psi(\beta)= \frac{3}{\beta} + \frac{K_1(\beta)}{K_2(\beta)}.
 \end{equation}
\begin{lemma}
\label{restmoments}
Let $J=J(n,\beta,0;\q)$. Then, the following equalities are verified:
\begin{enumerate}
\item
$\displaystyle
\int_{\R^3} J \ d\q = n,
$
\item
$\displaystyle
\int_{\R^3} q^i J \ d\q = \int_{\R^3} q^i J \frac{d\q}{q^0
} = 0,
$
\item
$\displaystyle
\int_{\R^3} |\q|^2  J \ \frac{d\q}{q^0
} = \frac{3 n}{\beta},
$
\item
$\displaystyle
\int_{\R^3} J \frac{d\q}{q^0
} = \frac{n}{M(\beta)} \frac{4 \pi}{\beta} K_1(\beta)= n \frac{K_1(\beta)}{K_2(\beta)},
$
\item
$\displaystyle
\int_{\R^3} \sqrt{1 + |\q|^2} J \ d\q = n \Psi(\beta).
$
\end{enumerate}
\end{lemma}
\begin{proof}
The first relation follows from the very definition of $M(\beta)$, and the second one just by a symmetry argument. To obtain the third one, we use integration by parts:
\begin{eqnarray}
\nonumber
  \int_{\R^3}  |\q|^2 J \  \frac{d\q}{q^0
} &&=  4 \pi \int_0^\infty \frac{r^4}{\sqrt{1 + r^2}} J \, dr = \frac{4 \pi n}{M(\beta)}\int_0^\infty \frac{r^4}{\sqrt{1 + r^2}} e^{-\beta \sqrt{1+ r^2}} \, dr \nonumber
\\
&&\hspace{-1 cm} = -\frac{4 \pi n}{\beta M(\beta)} \int_0^\infty \frac{d}{dr}\left(e^{-\beta \sqrt{1+r^2}} \right) r^3 \, dr = \frac{12 \pi n}{\beta M(\beta)} \int_0^\infty r^2 e^{-\beta \sqrt{1+ r^2}} \, dr
\nonumber
\\
&&\hspace{-1 cm} = \frac{3 n}{\beta M(\beta)} \int_{\R^3} e^{-\beta \sqrt{1 + |\q|^2}}\, dq = \frac{3n}{\beta}.
\nonumber
\end{eqnarray}
The fourth relation is a consequence of the following identity
$$
   \int_{\R^3} \frac{e^{-\beta \sqrt{1 + |\q|^2}}}{\sqrt{1 + |\q|^2}}\, d\q = 4 \pi \int_0^\infty r^2 \frac{e^{-\beta \sqrt{1 + r^2}}}{\sqrt{1 + r^2}}\, dr = \frac{4 \pi}{\beta}K_1(\beta).
$$
The sum of the third and fourth relations yields the fifth one by using \eqref{defpsi}.
\end{proof}

The moments of the J\"uttner distribution in general form can be obtained thanks to the following classical decomposition (see \cite{LandauFluid} for instance):
\begin{lemma}
\label{lema4}
The energy momentum tensor $T^{\mu
\nu}$ can be expressed as:
\begin{equation}
\nonumber
 T^{\mu \nu}= \int_{\R^3} q^\mu q^\nu f \frac{d\q}{ q^0
}=-p_{f}
g^{\mu \nu}+ (e_{f}+p_{f})u^\mu u^\nu.
\end{equation}
\end{lemma}
Then all the moments of a given $f$ that appear as components of $T^{\mu \nu}$ can be computed once we have the values of $p_f$ and $e_f$.
This can be combined with Lemma \ref{lorentz}, which ensures that it suffices to compute the local energy and pressure in the Lorentz rest frame. These two are given by formulae \eqref{efacil} and \eqref{pfacil}.

For the special case of $f=J$ we can go further as the computations in formulae \eqref{efacil} and \eqref{pfacil} were already carried in Lemma \ref{restmoments}. Then we get
the following result.
\begin{lemma}
\label{average_energy}
The quantities $e_{J}$ and $p_{J}$ are given by
\begin{equation}
\nonumber
e_{J}=n\Psi(\beta), \quad p_{J}= \frac{n}{\beta}.
\end{equation}
Using the standard physical units,
\begin{equation}
\nonumber
e_{J}=c^2 n\Psi(\beta), \quad p_{J}= c^2 \frac{n}{\beta}.
\end{equation}
\end{lemma}
A direct application of the program sketched above yields then
\begin{lemma}
\label{lema3}
Given any J\"uttner distribution $J$, the following relations hold true:
\begin{enumerate}
\item
$\displaystyle
\int_{\R^3} q^\mu J \frac{d\q}{q^0} = n u^\mu,
$
\item
$\displaystyle
\int_{\R^3}  |\q|^2 J \ \frac{d\q}{q^0} = e_{J} |\u|^2  + p_{J} (3 + |\u|^2) =n \Psi(\beta)|\u|^2 + (3 + |\u|^2)\frac{n}{\beta},
$
\item
$\displaystyle
\int_{\R^3} \sqrt{1 + |\q|^2} J \ d\q = p_{J} |\u|^2 + e_{J} (1 + |\u|^2) = \frac{n}{\beta}|\u|^2 + n \Psi(\beta) (1 + |\u|^2),
$
\item
$\displaystyle
\int_{\R^3} q^i J d\q = (e_{J} +p_{J})\sqrt{1 + |\u|^2} u^i = \left(n \Psi(\beta) + \frac{n}{\beta} \right) \sqrt{1 + |\u|^2} u^i,
$
\item
$\displaystyle
\int_{\R^3} J \frac{d\q}{q^0}  = e_{J} - 3 p_{J}= n \left( \Psi(\beta) - \frac{3}{\beta} \right) = n \frac{K_1(\beta)}{K_2(\beta)}.
$
\end{enumerate}
\end{lemma}
\begin{proof}
The first point follows from the  definition of $u^\mu$ in terms of $N^\mu$. For the remaining ones,
we just take recourse on Lemma \ref{lema4}. From there we get that
\begin{eqnarray}
\nonumber
\int_{\R^3} q^0
 J \ d\q &=&  T^{00} = - p_{J} + (e_{J}+p_{J})(1+|\u_{J}|^2),
\nonumber
\\
\int_{\R^3} q^i J \ d\q &=& T^{i0} = (e_{J}+p_{J}) u_{J}^i \sqrt{1 +|\u_{J}|^2},
\nonumber
\\
 \int_{\R^3} \frac{|\q|^2}{q^0
} J \ d\q &=& T^{11}+T^{22}+T^{33} = 3 p_{J} + (e_{J}+p_{J}) |\u_{J}|^2.
 \nonumber
\end{eqnarray}
This is to be combined with Lemma \ref{average_energy}. To conclude, we notice that
$$
  \int_{\R^3} J \frac{d\q}{q^0
} = \int_{\R^3} \left((q^0)
^2 - |\q|^2\right) J \frac{d\q}{q^0}
$$
and the last relation follows.
\end{proof}

\subsection{Entropy fluxes}

We can compute also the entropy densities and fluxes of the J\"uttner
 equilibrium, which are used to obtain information in the hydrodynamical limit, thanks to the H-theorem.
\begin{lemma}
\label{sieben}
 The following equalities are verified:
\begin{eqnarray}
 \label{cuarentaycincoa}
\int_{\R^3} \frac{q^i \q \cdot \u }{q^0} J d\q &=&  u^i (p_{J}+(e_{J}+p_{J})|\u|^2),
\\
\int_{\R^3}J \ln(J)\ d\q &=&  n\sqrt{1+|\u|^2}\bigg(\ln \left(
\frac{n}{ M(\beta)}\right)-\beta \Psi(\beta)\bigg),
\label{cuarentaycinco}
\\
\int_{\R^3}\frac {q^i}{q^0}J\ln(J)\ d\q &=& nu^i\bigg( \ln
\left(\frac{n}{M(\beta)}\right)- {\bf \beta }\Psi(\beta) \bigg).
 \nonumber
\end{eqnarray}
\end{lemma}
\begin{proof}  Using Lemma \ref{lema4} we get
$$
\int_{\R^3} \frac{q^i \q \cdot \u}{q^0
}J d\q =  u_j (p_{J}
\delta^{ij} +(e_{J}+p_{J})u^i u^j)= u^i(p_{J}+(e_{J}+p_{J})|\u|^2).
$$
To prove (\ref{cuarentaycinco}), note that

\begin{eqnarray}
\int_{\R^3}J\ln(J)\ d\q &=& \ln \left(\frac{n}{M(\beta)}\right)\int_{\R^3}J\ d\q -\beta \sqrt{1+|\u|^2}\int_{\R^3}
\sqrt{1+|\q|^2} \, J \, d\q \nonumber
\\
&+ & \beta \u\cdot
\int_{\R^3}\q J\ d\q \nonumber
\end{eqnarray}
and then using Lemma \ref{lema3}, items {\it 1,3} and {\it 4} we obtain (\ref{cuarentaycinco}).

In the same way
\begin{eqnarray}
\int_{\R^3}\frac{\q}{q^0}J\ln(J)\ d\q &=& \ln \left(\frac{n}{M(\beta)}\right)\int_{\R^3} \frac{\q}{q^0}J d\q -\beta
\sqrt{1+|\u|^2}\int_{\R^3} \q \, J \, d\q \nonumber
\\
&+ & \beta  \int_{\R^3}\frac{\q}{q^0} \q\cdot \u J\ d\q, \nonumber
\end{eqnarray}
and using Lemmas \ref{average_energy} and \ref{lema3}, items {\it 1} and {\it 4}, combined with \eqref{cuarentaycincoa} we arrive to the last identity. This proves the Lemma.
\end{proof}
\subsection{Monotonicity of $K_1/K_2$}\label{mono}
To begin with, let us recall the following recurrence relation:
\begin{equation}
\label{rec}
K_2(\beta)= \frac{2}{\beta} K_1(\beta) + K_0(\beta).
\end{equation}
This can be used to show that
$$
  \frac{K_2(\beta)}{K_1(\beta)} \le \frac{2}{\beta}+1,
$$
as $K_0(\beta) < K_1(\beta)$.
We also note that
\begin{eqnarray}
\label{kedo}
 \left(\frac{K_1(\beta)}{K_2(\beta)}\right)' &=& \frac{3}{\beta} \frac{K_1(\beta)}{K_2(\beta)} + \left(\frac{K_1(\beta)}{K_2(\beta)}\right)^2-1 \\
\nonumber 
 &\ge& \frac{3}{\beta+2}+\frac{\beta^2}{(\beta+2)^2}-1= \frac{2-\beta}{(\beta+2)^2},
\end{eqnarray}
which is strictly positive for $\beta<2$.

Next we analyze the case $\beta \ge 2$. For that we deal with the integral representations of the incomplete Bessel functions. Using the substitution $x=\sinh(s/2)$ we get
$$
K_0(\beta)+K_1(\beta)= \int_0^\infty (1+\cosh(s)) e^{-\beta \cosh(s)}\ ds = 2 \int_0^\infty \frac{2+2x^2}{\sqrt{1+x^2}} e^{-\beta(1+2x^2)}\ dx.
$$
By means of the inequality\footnote{We use that the binomial series is alternate and the fact that when we truncate the series, the error term has the same sign as the first term that is discarded.} 
$$
\frac{1}{\sqrt{1+x^2}}\ge 1-\frac{x^2}{2}\quad \mbox{for}\ x>0
$$ 
we obtain the estimate
$$
K_0(\beta)+K_1(\beta) \ge e^{-\beta} \frac{\sqrt{2 \pi}}{\sqrt{\beta}} \left(1+\frac{1}{8 \beta}-\frac{3}{32 \beta^2} \right).
$$
Arguing in a similar way but using this time the inequality
$$
\frac{1}{\sqrt{1+x^2}} \le 1-\frac{x^2}{2} + \frac{3}{8}x^4 \quad \mbox{for}\ x>0
$$
we arrive to 
$$
K_0(\beta) \le e^{-\beta} \frac{\sqrt{2 \pi}}{2\sqrt{\beta}} \left(1-\frac{1}{8 \beta}+\frac{9}{128 \beta^2} \right).
$$
Therefore
$$
\frac{K_0(\beta) + K_1(\beta)}{K_0(\beta)} \ge \frac{256 \beta^2+32 \beta -24}{128 \beta^2 -16 \beta + 9}
$$
and
$$
\frac{K_1(\beta)}{K_0(\beta)} \ge \frac{128 \beta^2+48 \beta -33}{128 \beta^2 -16 \beta + 9}.
$$
Notice that both the numerator and the denominator are positive for $\beta \ge 2$. This estimate can be used in combination with \eqref{rec} to get
$$
\frac{K_2(\beta)}{K_1(\beta)} \le \frac{128 \beta^3+240 \beta^2+105 \beta -66}{128 \beta^3+48 \beta^2-33 \beta}.
$$
We plug this into \eqref{kedo} so that
$$
\left(\frac{K_1(\beta)}{K_2(\beta)} \right)' \ge  3 \frac{128 \beta^2 + 48 \beta-33}{128 \beta^3+240 \beta^2+105 \beta -66} + \frac{(128 \beta^2 +48 \beta -33)^2 \beta^2}{(128 \beta^3+240 \beta^2+105 \beta -66)^2}-1
$$
$$
 = \frac{3 (6656 \beta^4 + 8512 \beta^3 - 4080 \beta^2 -2013 \beta +726)}{(128 \beta^3+240 \beta^2+105 \beta -66)^2}.
$$
Using that $\beta \ge 2$ we conclude with
$$
\left(\frac{K_1(\beta)}{K_2(\beta)} \right)' \ge \frac{3 (6656 \beta^4 + 2419 \beta^3  +726)}{(128 \beta^3+240 \beta^2+105 \beta -66)^2} >0.
$$

%

{\bf ACKNOWLEDGEMENTS}

The authors thank Prof. Bert Janssen for fruitful discussions that helped us to improve the contents of this paper.  
 This work was partially supported by Ministerio de Ciencia e
Innovaci\'on (Spain), project MTM2011-23384.  The first
author was {supported by Hassan II Academy of Sciences and
Technology (Morocco)}. The second author is partially supported by a Juan de la Cierva grant of the spanish MEC.

 \end{document}